\documentclass[
reprint,
superscriptaddress,
 amsmath,amssymb,
 aps,
twocolumn,
prx,
floatfix,
noeprint
]{revtex4-2}

\usepackage{graphicx}
\usepackage{dcolumn}
\usepackage{bm}
\usepackage{braket}
\usepackage{amsmath}
\usepackage{textcomp}
\usepackage{textgreek}
\usepackage{hyperref}
\hypersetup{colorlinks=true, citecolor=blue, urlcolor=blue, linkcolor=blue}

\usepackage{mathtools}

\usepackage{amsfonts}
\usepackage[bb=boondox]{mathalfa}
\usepackage{braket}
\usepackage{dsfont}
\usepackage{xcolor}
\usepackage{multirow}

\DeclareMathOperator{\Tr}{Tr}

\DeclareMathOperator{\sign}{sign}

\usepackage{amsthm}
\newtheorem{theorem}{Theorem}
\newtheorem{lemma}[theorem]{Lemma}
\newtheorem{problem}[theorem]{Problem}
\newtheorem{definition}[theorem]{Definition}

\begin{document}


\title{Machine learning of quantum data using optimal similarity measurements}


\author{Zhenghao Li}\thanks{These authors contributed equally to this work.}
\affiliation{Department of Physics, Imperial College London, Prince Consort Road, London, SW7 2AZ, UK}

\author{Hao Zhan}\thanks{These authors contributed equally to this work.}
\affiliation{Department of Physics, Imperial College London, Prince Consort Road, London, SW7 2AZ, UK}
\affiliation{College of Engineering and Applied Sciences, Nanjing University, 163 Xianlin Road, Nanjing 210093, China}

\author{Shana H. Winston}
\affiliation{Department of Physics, Imperial College London, Prince Consort Road, London, SW7 2AZ, UK}

\author{Ewan Mer}
\affiliation{Department of Physics, Imperial College London, Prince Consort Road, London, SW7 2AZ, UK}

\author{Zhenghao Yin}
\affiliation{University of Vienna, Faculty of Physics, Vienna Center for Quantum Science and Technology (VCQ), Boltzmanngasse 5, Vienna A-1090, Austria}

\author{Shang Yu}
\affiliation{Department of Physics, Imperial College London, Prince Consort Road, London, SW7 2AZ, UK}
 \affiliation{Centre for Quantum Engineering, Science and Technology (QuEST),
Imperial College London, Prince Consort Rd, London, SW7 2AZ, United Kingdom}

\author{Yazeed K. Alwehaibi}
\affiliation{Department of Physics, Imperial College London, Prince Consort Road, London, SW7 2AZ, UK}

\author{Gerard J. Machado}
\affiliation{Department of Physics, Imperial College London, Prince Consort Road, London, SW7 2AZ, UK}
\affiliation{Clarendon Laboratory, University of Oxford, Parks Road, Oxford OX1 3PU, UK}

\author{Dayne Marcus Lopena}
\affiliation{Department of Physics, Imperial College London, Prince Consort Road, London, SW7 2AZ, UK}
 \affiliation{Centre for Quantum Engineering, Science and Technology (QuEST),
Imperial College London, Prince Consort Rd, London, SW7 2AZ, United Kingdom}

\author{Lijian Zhang}
\affiliation{College of Engineering and Applied Sciences, Nanjing University, 163 Xianlin Road, Nanjing 210093, China}

\author{M. S. Kim}
\email[]{m.kim@imperial.ac.uk}
\affiliation{Department of Physics, Imperial College London, Prince Consort Road, London, SW7 2AZ, UK}
 \affiliation{Centre for Quantum Engineering, Science and Technology (QuEST),
Imperial College London, Prince Consort Rd, London, SW7 2AZ, United Kingdom}

\author{Aonan Zhang}
\email[]{aonan.zhang@physics.ox.ac.uk}
\affiliation{Department of Physics, Imperial College London, Prince Consort Road, London, SW7 2AZ, UK}
\affiliation{Clarendon Laboratory, University of Oxford, Parks Road, Oxford OX1 3PU, UK}

\author{Ian A. Walmsley}
\email[]{ian.walmsley@physics.ox.ac.uk}
\affiliation{Department of Physics, Imperial College London, Prince Consort Road, London, SW7 2AZ, UK}
 \affiliation{Centre for Quantum Engineering, Science and Technology (QuEST),
Imperial College London, Prince Consort Rd, London, SW7 2AZ, United Kingdom}
\affiliation{Clarendon Laboratory, University of Oxford, Parks Road, Oxford OX1 3PU, UK}

\author{Raj B. Patel}
\email[]{raj.patel1@imperial.ac.uk}
\affiliation{Department of Physics, Imperial College London, Prince Consort Road, London, SW7 2AZ, UK}
 \affiliation{Centre for Quantum Engineering, Science and Technology (QuEST),
Imperial College London, Prince Consort Rd, London, SW7 2AZ, United Kingdom}


\begin{abstract}
Quantum machine learning seeks a computational advantage in data processing by evaluating functions of quantum states, such as their similarity, that can be classically intractable to compute. For quantum advantage to be possible, however, it is essential to bypass costly characterisation of individual data instances in favour of efficient, direct similarity evaluation. Here we demonstrate a sample-optimal, hardware-efficient protocol for estimating quantum similarity -- the state overlap -- using bosonic quantum interference. The sample complexity of this approach is independent of the system dimension and is information-theoretically optimal up to a constant factor. Experimentally, we implement the scheme on \emph{Prakash-1}, a quantum computing platform based on a fully programmable integrated photonic processor. By preparing and interfering qudit states on the chip to directly extract their overlap, we demonstrate classification and online learning of quantum data with high accuracy in realistic noisy experiments. Our results establish joint overlap measurements as a scalable pathway to efficient quantum data analysis and a practical building block for network-integrated quantum machine learning.
\end{abstract}

\maketitle

\section{Introduction}

Assessing the similarity between objects is an instinctive process fundamental to human cognition. Mirroring this, the computation of similarity drives modern machine learning~\cite{bishop2006pattern, hastieElementsStatisticalLearning2009}, and more recently, guides attention in transformer-based large language models through vector inner products~\cite{vaswani2017attention,Brown2020}.
Quantum machine learning (QML) extends this concept to an exponentially large Hilbert space, where the overlap between two quantum states---the squared inner product between state vectors---plays the same organising role, serving, for instance, as the kernel function in quantum classifiers or the cost function in training quantum models (Fig.~\ref{fig: conceptual})~\cite{Biamonte_2017_QML,Havl_ek_2019_supervised_learning,Rebentrost2014_quantum_svm_big_data,Jerbi_2023_QML_beyond_kernel,Cerezo_2021_VQA,Beer2020_training_deep_qnn}. Evaluating the quantum state overlap is a BQP-complete problem~\cite{rethinasamyEstimatingDistinguishabilityMeasures2023}, which forms the computational foundation for potential quantum advantage in machine learning~\cite{Liu_2021_rigorous_speedup, Huang_2021_power_of_data}. 

Unlocking this quantum advantage relies on efficient measurement of similarities between quantum data at scale. Quantum mechanics enables direct processing of joint information between two unknown states, departing from the classical intuition that comparing objects first requires characterising them individually. The latter is challenging in emerging quantum networks, where data are natively generated and transmitted as quantum states, such as the outputs of quantum simulators, processors, and sensors (Fig.~\ref{fig: conceptual}a). While quantum kernels and variational learners have recently been implemented across various platforms~\cite{Peters_2021_ml_high_dim,Abbas2021_power_of_qnn,Pan2023_deep_qnn_superconducting,Yin_2025_experimental_kernel,Hoch_2025}, scaling up these protocols remains constrained by the unfavourable cost of overlap measurements. Protocols based on tomographic reconstruction or distributed measurements of individual data instances generally require a number of samples that grows exponentially with system size~\cite{Haah-2016-optimal_tomo,Anshu-2022-distributed_quantum_inner_product, Zhan_2025_experimental_benchmarking}. Alternative strategies attempt to bypass this complexity by encoding classical data into quantum circuits and concatenating one unitary with the inverse of another, but the experimental overhead of implementing deep unitaries becomes prohibitive at large scale~\cite{Shende_2005}. Moreover, circuit-encoding methods require \textit{a priori} knowledge of the classical description of the data, rendering them inapplicable to quantum-native data increasingly targeted by QML~\cite{Cerezo_2022_challenges_qml}. Direct sample- and hardware-efficient similarity evaluation becomes essential for scalable QML.

\begin{figure*}
    \includegraphics[width=\textwidth]{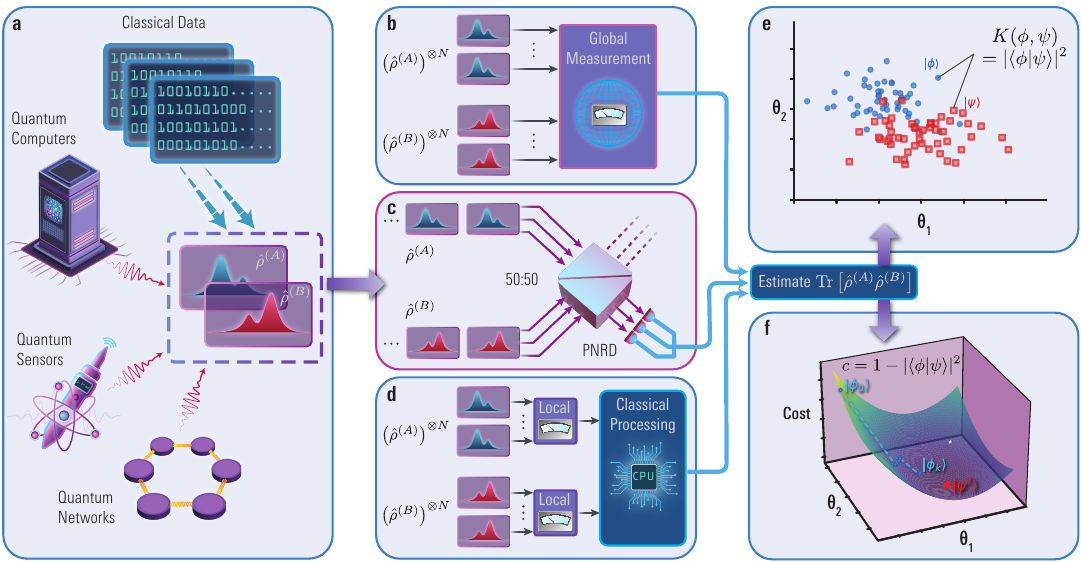}
    \caption{\label{fig: conceptual} Quantum machine learning (QML) with similarity measurements. (a) QML can process quantum states that encode classical data, though it also shows potential in learning native quantum data, such as the output states from fault-tolerant quantum computers, analogue quantum simulators, quantum sensors or information passed through large-scale quantum networks. A hierarchy of measurements exists to measure the similarity between two unknown quantum states via: (b) Global measurement on $N$ copies of each state; (c) Sequential joint measurement using multi-mode bosonic interference on single copy of each state, derived and demonstrated in this work; (d) Distributed local measurements on $N$ copies of each state separately followed by classical post-processing, an example of which we derive for CV states. Within this hierarchy, the interference-based joint measurement combines optimal sample complexity with experimental accessibility. For application in QML, we demonstrate the measurement as evaluation of (e) kernel functions for quantum data classification and (f) cost function for online learning. PNRD: photon-number-resolving detector.}
\end{figure*}

In this study, we establish a sample-optimal and scalable overlap estimation scheme for arbitrary photonic states based on bosonic quantum interference and photon-number parity readout. Unlike recent entanglement-assisted protocols, designed to learn individual quantum states or channels~\cite{Huang-2022-quantum_advantage_learning_from_experiments, Oh_2024_random_displacement_channel, Liu_2025_advantage_scalable, coroi2025exponentialadvantage}, our work exploits entangled measurements to directly learn the joint properties between two states. In addition to its experimental simplicity, the scheme provably achieves the optimal $O(\epsilon^{-2})$ sample complexity for an $\epsilon$-additive error precision. The scaling is optimal up to a constant factor across all possible measurement strategies, even when multi-copy global measurements are considered, whereas any distributed local measurement likely incurs an exponential cost. Crucially, the sample complexity is independent of the mode number, photon number or energy, thereby permitting scalable processing of high-dimensional quantum data.

Photons are ideal carriers of quantum data in future quantum networks, offering low decoherence, room-temperature operation, and compatibility with telecom infrastructure, along with the rich capability of encoding both discrete-variable (DV) states and continuous-variable (CV) states. In photonic systems, similarity measurements have been previously derived for single-mode Gaussian states by linear-optical interference and Wigner function estimation~\cite{Kim-2002-experimentally}, and later for DV states using two-photon interference as ancilla-free SWAP tests~\cite{GarciaEscartin2013_swap_test_HOM, Patel_2016, Zhang_2021_np_problems, Zhan_2025_experimental_benchmarking} with recent generalisation to CV states~\cite{volkoff-2022-AncillafreeContinuousvariableSWAP}. Building on these foundations, we show that bosonic interference provides a universal and hardware-native solution for measuring similarities between arbitrary multi-photon multi-mode states with optimal information-theoretic complexity, highlighting a hierarchy in information extraction where joint measurements may have an exponential advantage over distributed strategies.

Experimentally, we validate the measurement scheme and its application to QML on \emph{Prakash-1}, a photonic quantum computing platform with a fully programmable photonic integrated circuit (PIC) at its core, which implements a state-of-the-art compact rectangular decomposition of universal linear optics~\cite{bell-2021-further_compact}. We prepare phase-encoded multi-mode qudit states and directly measure their overlap from the photon statistics. With experimentally estimated overlaps, we perform two QML tasks: classification of quantum datasets with over 90\% accuracy and online learning of unknown target data with 98.3\% median fidelity. Our results establish the joint overlap measurement via quantum interference as a hardware-efficient, sample-optimal enabler for scalable quantum data processing. 

\section{Overlap measurement scheme}

Photonic quantum data can be DV states, such as $M$-mode $K$-photon Fock states with Hilbert space dimension $d={M+K-1 \choose K}$, or CV states with theoretically infinite Hilbert space dimension that is only bounded by energy, including Gaussian states or non-Gaussian states such as Gottesman-Kitaev-Preskill (GKP) states~\cite{gottesmanEncodingQubitOscillator2001}. We consider two general $M$-mode states $\hat{\rho}^{(A)}$ and $\hat{\rho}^{(B)}$ with their trace overlap given by $\Tr\left[\hat{\rho}^{(A)} \hat{\rho}^{(B)}\right]$, which reduces to the squared overlap $\left|\langle \psi^{(A)}|\psi^{(B)}\rangle\right|^2$ for pure states $\hat{\rho}^{(A/B)}=|\psi^{(A/B)}\rangle\langle\psi^{(A/B)}|$.

Given $N$ copies of each state, we consider three types of measurement schemes for estimating the overlap: 
\begin{enumerate}
    \item Global measurements on all $2N$ state copies, $\left(\hat{\rho}^{(A)}\otimes \hat{\rho}^{(B)}\right)^{\otimes N}$;
    \item Joint measurement via bosonic interference on only one copy of each state at a time, $\left(\hat{\rho}^{(A)}\otimes \hat{\rho}^{(B)}\right)$;
    \item Distributed measurements that perform local measurements on $\left(\hat{\rho}^{(A)}\right)^{\otimes N}$ and $\left(\hat{\rho}^{(B)}\right)^{\otimes N}$ respectively followed by post-processing of the outcomes. 
\end{enumerate}
The three measurement strategies are visualised in Fig.~\ref{fig: conceptual}b-d. Both global and distributed schemes can allow arbitrary, adaptive, multi-round measurements.

Any $M$-mode photonic state can be described by $\hat{\rho} = \frac{1}{\pi^M} \int_{\mathds{C}^M} d^{2M}\bm{\alpha} \chi(\bm{\alpha}) \hat{D}(-\bm{\alpha})$, where $\chi$ is its characteristic function and  $\hat{D}$ is the $M$-mode displacement operator. The overlap between two states, $\hat{\rho}^{(A/B)}$, is given by,
\begin{equation}\label{eqn: overlap definition}
    \Tr\left[\hat{\rho}^{(A)} \hat{\rho}^{(B)}\right]
    =
    \frac{1}{\pi^M} \int_{\mathds{C}^M} d^{2M}\bm{\alpha} \chi^{(A)}(\bm{\alpha}) \chi^{(B)} (-\bm{\alpha}).
\end{equation}

The two states, originating in registers labelled by $A/B$, are interfered on balanced beamsplitters (BSs). The BSs are described by a unitary operator $\hat{U}_{\textup{BS}}$ that combines the modes pairwise between the two registers: $\hat{U}_{\textup{BS}}^\dagger \hat{a}^{(A)}_{i}\hat{U}_{\textup{BS}} = \frac{1}{\sqrt{2}}\left(\hat{a}^{(A)}_i + \hat{a}^{(B)}_i \right), 
\hat{U}_{\textup{BS}}^\dagger \hat{a}^{(B)}_{i}\hat{U}_{\textup{BS}} = \frac{1}{\sqrt{2}}\left(\hat{a}^{(A)}_i - \hat{a}^{(B)}_i \right)$, where the $i$-th mode in a register is labelled by subscript $i=1,\hdots,M$. 
We then measure the photon-number parity, $\hat{\Pi} = (-1)^{\sum_{i=1}^M\hat{n}_i^{(B)}}$, on output register $B$, where $\hat{n}_i^{(B)}$ is the photon number operator in the $i$-th mode. This gives a direct estimate of the two-state overlap:
\begin{equation}\label{eqn: overlap measurement by parity}
    \Tr\left[\hat{\rho}^{(A)} \hat{\rho}^{(B)}\right] 
    = 
    \Tr_{(B)}\left[
    \Tr_{(A)} \left(\hat{U}_{\textup{BS}}^\dagger \hat{\rho}^{(A)} \hat{\rho}^{(B)} \hat{U}_{\textup{BS}} \right) \hat\Pi \right].
\end{equation}
The detailed derivation of Equation~\ref{eqn: overlap measurement by parity} is provided in Methods and Supplementary Information. While Equation~\ref{eqn: overlap measurement by parity} holds true for any general mixed states, we note that the trace overlap is a good measure of similarity when at least one of two states is pure.

Operationally, we can measure the output state using photon-number-resolving detectors (PNRDs). We sample $N$ shots and record the number of even- and odd-total photon number events, from which we construct an unbiased overlap estimator. Using Hoeffding's bound, we can prove its sample complexity in Theorem~\ref{thm: sample complexity joint overlap estimation} below. 

\begin{theorem}\label{thm: sample complexity joint overlap estimation}
    Given $N$ copies of $\hat{\rho}^{(A)} \otimes \hat{\rho}^{(B)}$ and use of balanced BSs and PNRDs, the overlap $\Tr\left[\hat{\rho}^{(A)} \hat{\rho}^{(B)}\right]$ can be estimated to within additive error $\epsilon$ with success probability of at least $1-\delta$ for sample complexity $N=O\left(\epsilon^{-2}\log(2\delta^{-1})\right)$. 
\end{theorem}

The sample complexity simplifies to $N=O(\epsilon^{-2})$ for a constant success probability $1-\delta>1/2$. The full proof is provided in the Supplementary Information. Importantly, the sample complexity is independent of the mode number, photon number, or energy, and requires only a single copy of each state per measurement. In the next section, we show that this performance is not only efficient, but also optimal within the full hierarchy of measurement strategies. 

\section{Sample complexity optimality}

\begin{table*}[htbp]
\caption{\label{tab: complexity comparison} Comparison of sample complexity for joint overlap estimation and distributed overlap estimation of discrete-variable (DV) states and continuous-variable (CV) states, respectively, to additive error $\epsilon$ and success probability $1-\delta$. We summarise the information-theoretic lower bounds for comparison with the schemes derived in this work. Symbol $d$ denotes dimension number of DV states, $M$ mode number of photonic states. See main text and Methods for details.}
\renewcommand{\arraystretch}{1.5}
\begin{ruledtabular}
\begin{tabular}{cccc} 
 & & 
 \begin{tabular}[c]{@{}c@{}} Joint overlap estimation \\ on $(\hat{\rho}^{(A)}\otimes \hat{\rho}^{(B)})^{\otimes N}$ \end{tabular}
 & 
 \begin{tabular}[c]{@{}c@{}} Distributed overlap estimation \\ on  $(\hat{\rho}^{(A)})^{\otimes N}$, $(\hat{\rho}^{(B)})^{\otimes N}$ \end{tabular}
 \\ \hline
 \multirow{2}{*}{Lower bound} & DV states
 & $\Omega\left(\epsilon^{-2} (1-2\delta)^2/4\right)$ \cite{Anshu-2022-distributed_quantum_inner_product} 
 & $\Omega\left(\max(\epsilon^{-2}, \sqrt{d}\epsilon^{-1})\right)$~\cite{Anshu-2022-distributed_quantum_inner_product} 
 \\ 
  & CV states 
 & $\Omega\left(\epsilon^{-2} (1-2\delta)^2/4\right)$ 
 \footnote{Extended in our work to CV states from the proof of \cite{Anshu-2022-distributed_quantum_inner_product}.}
 & 
 Lower bound unknown 
 \\
 \hline 
 Our work & Any bosonic state & $O\left(\epsilon^{-2}\log(2\delta^{-1})\right)$
 & $\tilde{O}\left(\epsilon^{-4}(e\kappa)^{4M}/M^2\right)$ for self-reflective states 
 \footnote{Logarithmic terms ignored in big-$\tilde{O}$ notation for simplicity and $\kappa$ is a mode-extensive energy bound.}
\\ 
\end{tabular}
\end{ruledtabular}
\end{table*}

The significance of the $O(\epsilon^{-2})$ scaling is best understood by comparing it to the fundamental limit and other measurement strategies of overlap estimation. Using Helstrom's bound~\cite{helstrom_1969}, we can prove the following theorem. 
\begin{theorem}\label{thm: optimality of joint overlap estimation}
    Given the $N$-copy state, $\left(\rho^{(A)} \otimes \rho^{(B)}\right)^{\otimes N}$, if an algorithm estimates $\Tr\left[\rho^{(A)} \rho^{(B)}\right]$ to within additive error $\epsilon$ with success probability of at least $1-\delta$, then $N=\Omega\left(\epsilon^{-2} (1-2\delta)^2/4\right)$. 
\end{theorem}

The proof of Theorem~\ref{thm: optimality of joint overlap estimation} is first given in Ref.~\cite{Anshu-2022-distributed_quantum_inner_product} for the SWAP test in DV systems, which can be extended to arbitrary states, including CV states, with full proof provided in Supplementary Information. If given a constant success probability, such as $1-\delta = 2/3$, Theorems~\ref{thm: sample complexity joint overlap estimation} and \ref{thm: optimality of joint overlap estimation} provide a tight bound of $N=\Theta(\epsilon^{-2})$, proving that the interference-based measurement is already optimal up to a constant scaling, even when considering experimentally more complex global measurements.

We also motivate the optimality of our scheme by considering the distributed overlap estimation problem~\cite{Anshu-2022-distributed_quantum_inner_product} and argue that joint measurements likely have exponential advantage, as summarised in Table~\ref{tab: complexity comparison}. For the distributed problem, Alice and Bob are given $\left(\hat{\rho}^{(A/B)}\right)^{\otimes N}$ respectively and are required to estimate their overlap by only local quantum operations and classical communication channels, but with adaptive, multi-round measurements allowed.

For DV states, any distributed approach scales exponentially with system size. Full tomographic reconstruction of at least one state requires a sample complexity of $O(d)$~\cite{Haah-2016-optimal_tomo}, and distributed overlap estimation requires a sample complexity rigorously lower-bounded by $N=\Omega(\max(\epsilon^{-2}, \sqrt{d} \epsilon^{-1}))$~\cite{Anshu-2022-distributed_quantum_inner_product}. For an $M$-mode, $K$-photon state with Hilbert space dimension $d=\binom{M+K-1}{K}$, both complexities scale exponentially with the photon number and mode number.

For CV states with infinite density matrix dimension, it is often practical to assume some mode-extensive energy constraint, under which tomography is also inefficient with sample complexity being exponential in the mode number~\cite{mele2024learningquantumstatescontinuous}. For distributed CV overlap estimation, no rigorous lower bound exists to the best of our knowledge. However, since energy bounded CV states can be approximated by truncated DV states~\cite{mele2024learningquantumstatescontinuous}, we conjecture that the CV problem suffers from a similar exponential lower bound as its DV counterpart.

To explicitly illustrate the curse of dimensionality, we consider a special class of self-reflective states, which include many experimentally interesting CV states such as coherent states, squeezed vacuum states, cat states and GKP states~\cite{Wu-2024-efficient_learning}. For these states, distributed learning of \textit{individual} properties is `super-efficient': Only constant sample complexity is required to estimate their characteristic functions for an exponential number of points in phase space~\cite{Wu-2024-efficient_learning, coroi2025exponentialadvantage}, which can then be used to estimate the overlap in Equation~\ref{eqn: overlap definition} by Monte-Carlo integration.

However, the overall sample complexity is still inefficient. For an $M$-mode CV state, we also assume a mode-extensive energy constraint $\kappa$, such that the phase-space integral in Equation~\ref{eqn: overlap definition} is negligible for $|\bm{\alpha}| \geq \sqrt{\kappa M}$, thus defining a $2M$-dimensional hypersphere in phase space. Even though we can efficiently estimate the characteristic functions of two states, a constant precision in overlap estimation requires the precision of $\chi^{(A/B)}(\bm{\alpha}_i)$ at each point inside the hypersphere to scale inversely with the hypersphere volume, resulting in an overall sample complexity that remains exponential in $M$. We state the result informally in Theorem~\ref{theorem: informal distributed}.

\begin{theorem}\label{theorem: informal distributed}
    \textup{\textbf{(Informal.)}} 
    Given $N$ copies of two $M$-mode self-reflected states, $\hat{\rho}^{(A)}, \hat{\rho}^{(B)}$, whose characteristic functions are limited within a phase-space hypersphere with radius $|\bm{\alpha}| = \sqrt{\kappa M}$, the distributed overlap estimation problem can be solved with additive error $\epsilon$ and success probability of at least $2/3$ for sample complexity $N=\tilde{O}\left(\epsilon^{-4}(e\kappa)^{4M}/M^2\right)$. 
\end{theorem}

For notational simplicity, we drop logarithmic terms in the big-$\tilde{O}$ notation from the already exponential complexity. The formal statement and full proof are given in the Supplementary Information.

\section{Experimental implementation}

\begin{figure*}[!t]
    \centering
    \includegraphics[width=0.98\textwidth]{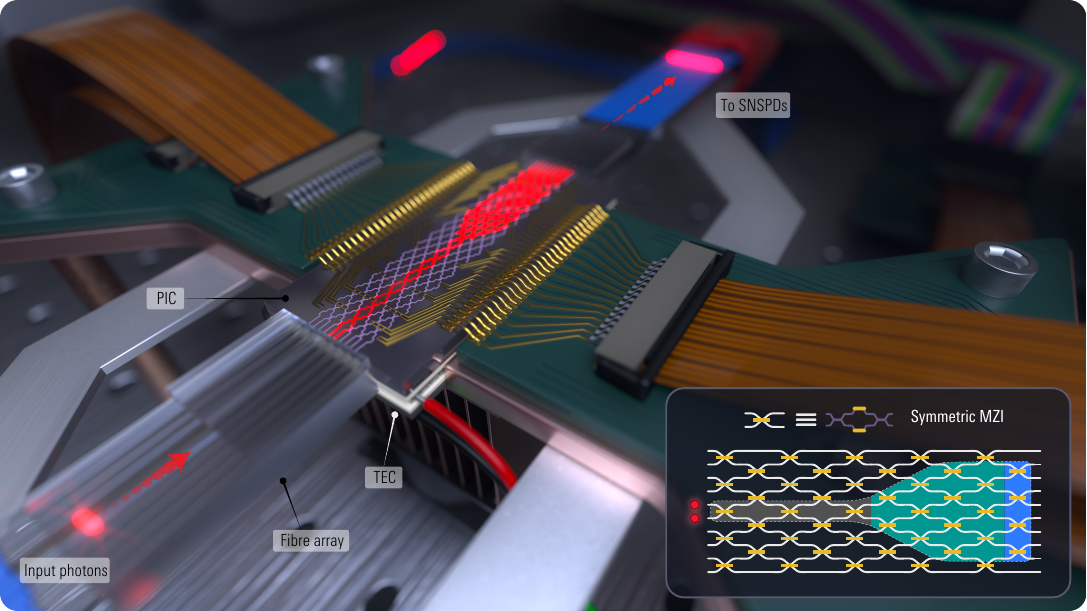}
    \caption{An illustration of the fully-programmable $\text{Si}_3\text{N}_4$ integrated photonic processor at the heart of the quantum computing platform \emph{Prakash-1}. A photonic integrated circuit (PIC), whose temperature is maintained by a thermo-electric cooling (TEC) unit, employs a square mesh of symmetric Mach-Zehnder interferometers (MZIs), with heating elements acting as thermo-optical phaseshifters in each arm of the MZIs. Collectively, the MZIs plus ten stand-alone on-chip phaseshifters (not used in this work and not shown) can implement an arbitrary $10\times10$ unitary transformation according to the Bell-scheme (see inset). The dashed line encloses the MZIs that are tuned in the experiment. The grey, teal, and blue areas indicate the MZIs used to route the photon pair, prepare the qudits, and perform the multi-mode interference, respectively. The output eight modes are fibre-coupled and routed to superconducting nanowire single-photon detectors (SNSPDs).
    }
    \label{fig:chip}
\end{figure*}

We experimentally demonstrate the joint overlap estimation scheme with application to QML routines. Our experimental set-up relies on spontaneous parametric downconversion (SPDC) for preparing an indistinguishable photon pair at $1550$~nm wavelength. The photon pair is produced by pumping a periodically-poled potassium titanyl phosphate (ppKTP) waveguide with 1~ps laser pulses at $775$~nm wavelength. 

A silicon nitride-based PIC implements a $10\times 10$ square mesh of tunable phaseshifters as shown in Fig.~\ref{fig:chip}, enabling universal linear optical operations. The PIC layout---designed in-house---implements a rectangular mesh of thermo-optic phaseshifters that can decompose arbitrary $10\times 10$ unitary transformations following the Bell-scheme in Ref.~\cite{bell-2021-further_compact}. The PIC is programmed to encode each photon into a qudit state that is a spatial superposition over four spatial modes with three tunable relative phases:
\begin{equation}\label{eqn: encoded qudit state}
    |\psi(\bm{\theta})\rangle 
    =
    \left(
    A_0 \hat{a}_{0}^\dagger + 
    \sum_{k=1}^3 A_k e^{i\theta_{k}}
    \hat{a}_{k}^\dagger
    \right) |0\rangle,
\end{equation}
where $\hat{a}_{k}^\dagger$ for $k\in[0,3]$ are the bosonic creation operators for the four spatial modes and $A_k$ are the real amplitudes. Vector $\bm{\theta}=[\theta_1, \theta_2,\theta_3]$ are the tunable parameters for either data encoding or model training. 

The two qudits occupy eight modes in total, which we split into two registers. According to Equation~\ref{eqn: overlap measurement by parity}, our general overlap estimation scheme requires four PNRDs on one of the registers. However, when the total photon number is at most two, an odd-parity measurement in one register is equivalent to a coincidence detection between the two registers, and the overlap is equal to one minus twice the probability of these coincidence events, $\langle \hat{\Pi} \rangle = 1 - 2 P_{\textup{odd}}$. Measurement in our experiment is performed by click-detection on all eight output modes by superconducting nanowire single-photon detectors (SNSPDs), which are non-photon-number-resolving detectors. 

\section{Classification of quantum data}

\begin{figure*}[th]
    \includegraphics[width=1.0\textwidth]{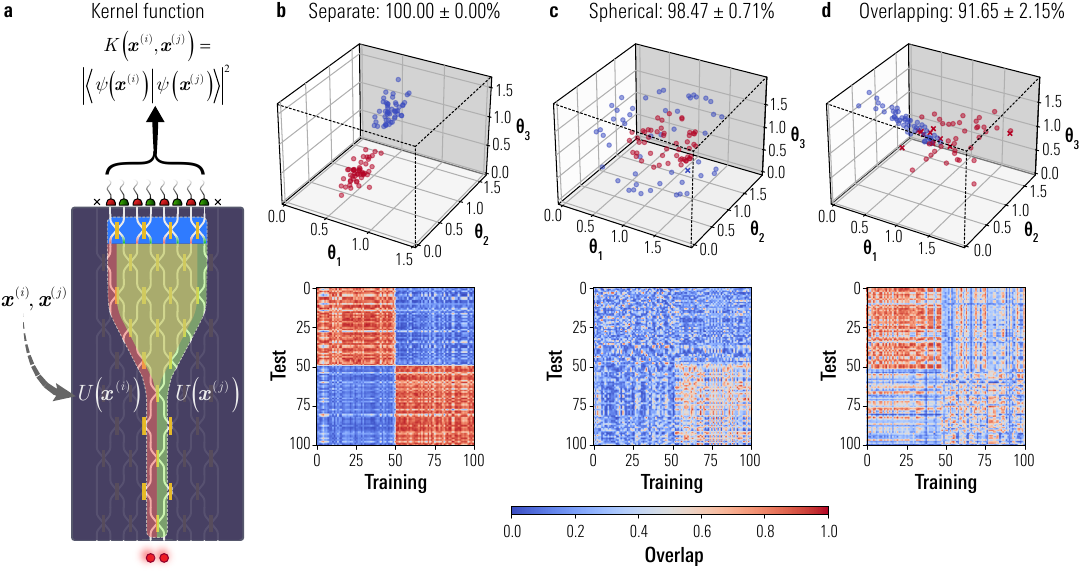}
    \caption{\label{fig: svm results} Quantum data classification. (a) Schematic of experimental routine for kernel function evaluation. (b-d) Classification results for linearly separable data, spherically separable data and overlapping data respectively. The top panels visualise the test datasets. Each dataset contains $100$ datapoints, which are three-dimensional vectors $\bm{x}=[\theta_1, \theta_2, \theta_3]$ represented by points in 3D space. The colour of each point indicates its groundtruth label, and is plotted with a circle marker if it is correctly classified, or a cross marker if it is misclassified. The bottom panels show the experimentally estimated overlap values (colour bar) between the test data and training data for each dataset. The datapoints are sorted by their labels for visualisation purposes.}
\end{figure*}

To demonstrate the utility of the optimal similarity measurement, we first implement a supervised kernel-based classification routine on quantum data using overlap-based kernel evaluation. Kernel methods play a foundational role in statistical learning theory, offering nonlinear decision boundaries with strong generalisation guarantees. Notably, training sufficiently wide neural networks has been found to be equivalent to kernel regression with the neural tangent kernel~\cite{Arthur2018_neural}.

In the quantum setting, kernel methods arise naturally: encoding data $\bm{x}^{(i)} \in \mathds{R}^d$ into quantum states $|\psi(\bm{x}^{(i)})\rangle$ defines a quantum feature map from the real space to Hilbert space, and the associated quantum kernel is evaluated by the state overlap, $K(\bm{x}^{(i)}, \bm{x}^{(j)}) = |\langle \psi(\bm{x}^{(i)}) | \psi(\bm{x}^{(j)}) \rangle|^2$. The task is to assign binary labels on unseen data $|\psi(\bm{x})\rangle$. We consider a setting in which training and test data are in the form of quantum states $|\psi(\bm{x})\rangle$ and all kernel entries $K(\bm{x}^{(i)}, \bm{x}^{(j)})$ are evaluated experimentally via bosonic interference. Given a set of training data $\{|\psi(\bm{x}^{(i)})\rangle\}_{i=1}^m$, each with label $y^{(i)}$, a support vector machine (SVM) constructs a linear classifier in the induced feature space, which predicts $\hat y(\bm{x})=\mathrm{sign}\!\left(\sum_{i=1}^m \beta^{(i)}y^{(i)}K(\bm{x},\bm{x}^{(i)})+b\right)$ for unseen data,
with coefficients $\{\beta^{(i)},b\}$ determined during training~\cite{hastieElementsStatisticalLearning2009} (see Methods for details).

In our implementation, we generate three labelled datasets of $200$ quantum states. Three-dimensional parameters are encoded in the qudit phases in Equation~\ref{eqn: encoded qudit state}, $\bm{x}^{(i)}=[\theta_1^{(i)}, \theta_2^{(i)}, \theta_3^{(i)}]$, as illustrated in Fig.~\ref{fig: svm results}a. The three datasets comprise linearly separable clusters, concentric spheres, and partially overlapping clusters in the original parameter space (Fig.~\ref{fig: svm results}). Each dataset is split randomly into $100$ training and $100$ test states. Using experimentally measured overlaps with $N=10^3$ samples per kernel evaluation, we obtain test accuracies of $100\%$ (separate), $98.47\pm0.71\%$ (spherical), and $91.65\pm2.15\%$ (overlapping), with uncertainties estimated by bootstrapping. The latter two datasets, which are not linearly separable in the original parameter space, become separable through the quantum kernel. While the datasets used here are not specifically designed to maximise the performance separation over classical algorithms~\cite{Huang_2021_power_of_data, Yin_2025_experimental_kernel}, the experiment demonstrates a scalable protocol for supervised learning in the Hilbert space using experimentally realistic resources.

\section{Online learning of quantum data}

\begin{figure*}[th]
    \centering
    \includegraphics[width=0.9\linewidth]{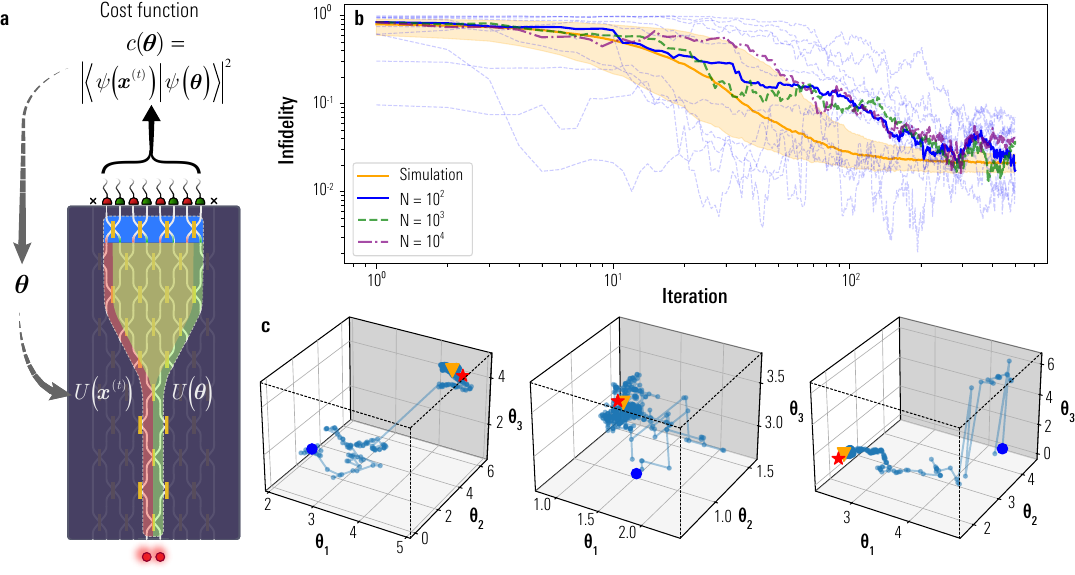}
    \caption{\label{fig: spsa results} Quantum data online learning. (a) Schematic of experimental routine for cost function evaluation. (b) Infidelity between the learned qudit state and the target qudit state against iteration number. Ten random target states are iteratively learned starting from random initial states, with $N=10^2$ state copies used per overlap evaluation (light blue dashed lines). The median infidelity across the ten learning tasks is shown by the dark blue solid line. The experiment is repeated for different number of state copies per overlap evaluation, $N=10^3$ and $N=10^4$, with median performance shown in green and purple dashed lines respectively. The shaded regions indicate the interquartile range (i.e. between the 25\% and 75\% quantiles) across simulations over $500$ randomised learning tasks, with median shown by the solid yellow line. (c) Optimisation paths for three example experiments in (b) with $N=10^2$, visualised in the three-dimensional space of qudit phases. The red star markers plot the target point, the blue circle markers plot the initial point of each path, and the yellow triangle markers plot the optimal point found after $500$ iterations. }
\end{figure*}

In contrast to batch learning, online learning receives input data sequentially and processes them on the fly. For quantum data, online learning is often practical as it is difficult to prepare a large batch of identical state copies on noisy quantum hardware~\cite{aaronsonOnlineLearningQuantum2019}. This scenario captures core tasks in QML, including training quantum neural networks~\cite{Beer2020_training_deep_qnn} and implementing quantum generative models such as quantum generative adversarial networks~\cite{Lloyd_2018_QGAN}, where a quantum model is trained to reproduce quantum data.

As a second application, we demonstrate overlap-based cost function evaluation in an online learning setting, where quantum data arrive sequentially as noisy copies of an unknown target state $|\psi(\bm{x}^{(t)})\rangle$. The objective is to train the output state $|\psi(\bm{\theta})\rangle$ of a parametrised quantum model with parameters $\bm{\theta}$ to approximate the target by minimising the infidelity $c(\bm{\theta}) = 1-|\langle\psi(\bm{x}^{(t)})|\psi(\bm{\theta})\rangle|^2$. In each experiment, a randomly initialised qudit is prepared on chip as the target state, while a second independently tunable qudit serves as the output of the trainable model. Parameters $\bm{\theta}$ are updated using simultaneous perturbation stochastic approximation (SPSA)~\cite{spall1998-overview-spsa}, with gradients estimated from two overlap evaluations per iteration (Fig.~\ref{fig: spsa results}a). This algorithm has been used in self-guided tomography for full state reconstruction, exhibiting superior efficiency and noise robustness compared to standard tomography methods~\cite{Ferrie-2014-SGQT, rambach-2021-self-guided, serino-2025-selfguided}. Thermal crosstalk between phaseshifters induces fluctuations in target preparation, mimicking realistic hardware noise in which successive copies are not identical.

Despite this noise, the optimisation converges reliably. Across ten random targets, the infidelity decreases to the few-percent level within $500$ iterations, with a median final value of $1.7\times10^{-2}$ (Fig.~\ref{fig: spsa results}b). At each iteration, the overlap evaluation consumes only $N=10^2$ copies of the target state. Moreover, increasing $N$ to $10^3$ or $10^4$ does not significantly change convergence rate or final precision, indicating robustness of the approach to shot noise. The convergence rate and final precision are in line with simulation results with a crosstalk noise model comparable with experimental parameters (see Supplementary Information for detailed results). Three representative optimisation trajectories in the three-dimensional parameter space of $\bm{\theta}$ are visualised in Fig.~\ref{fig: spsa results}c. These results highlight the robustness of the cost-function evaluation and its suitability for learning from quantum data in realistic, noisy settings.

\section{Discussion}

In certain problems of structured data processing, QML has shown provable advantage over classical methods~\cite{Liu_2021_rigorous_speedup,Huang_2021_power_of_data,Yin_2025_experimental_kernel}.
For efficient implementation, our optimal similarity measurement using bosonic interference achieves an exponential advantage over conventional methods that extract information about each state separately. In future quantum networks, where data are available naturally as quantum states, joint overlap estimation provides a minimal interface to extract relational information between quantum data, enabling applications such as distributed or blind quantum computing and learning~\cite{Knorzer2025distributedQIP, fitzsimonsPrivateQuantumComputation2017}..

The experimental demonstrations of QML on the \emph{Prakash-1} platform showcase the effectiveness of our approach in realistic, noisy experimental settings with modest quantum resources. For truly scalable QML, open questions of trainability in the form of barren plateaus or exponential concentration of kernel function values~\cite{thanasilp_exponential_2024, Larocca_2025_BP_in_variational} persist. Furthermore, the present implementation measures single-photon states evolved under linear-optical circuits, which have recently been shown to be classically simulable~\cite{lim-2025-efficientclassical}. The seek for generalisable end-to-end advantages in QML remains an open frontier. Nevertheless, we demonstrate a sample-optimal overlap estimation that addresses the measurement part of the challenge. The scheme is equally applicable to photonic states with non-Gaussianity and high entanglement that are expected to be classically intractable~\cite{konno2024logical,larsen2025integrated,yu2026extensible}. Scaling up the scheme and integrating such complex states are the next steps for achieving quantum advantage.

For future work, extending the two-state interference to multi-state interference, which has been experimentally explored in photonic systems~\cite{jones-2020-multiparticle}, could enable access to higher-order nonlinear functionals, expanding the expressive power of quantum models. Recent theoretical work has identified a hierarchy of learning capabilities enabled by increasing the number of jointly measured copies~\cite{noller-2025-infinite-hierarchy}. Higher-order functions between distinct quantum states could serve as the quantum analogue of multi-body attention mechanisms in transformer architectures, which potentially reduce hallucination and bias in generative artificial intelligence~\cite{huo-2025-capturing}. Looking ahead, the direct access to joint information between quantum objects enabled by quantum interference bypasses the costly need to characterise each object individually and opens new possibilities for scalable quantum intelligent networks.

\begin{acknowledgments}
The authors acknowledge Adam Taylor and Changhun Oh for useful discussions. R.B.P. thanks Bryn Bell for useful discussions regarding the PIC and acknowledges LioniX International for their service. MSK thanks Gyoyoung Jin for discussions on machine learning.
This work was supported by UK Research and Innovation Future Leaders Fellowship (project: MR/W011794/1) and Guarantee Postdoctoral Fellowship (projects: EP/Y029127/1, EP/Y029631/1), Engineering and Physical Sciences Research Council (EPSRC) UK Quantum Technologies Program's hubs for Quantum Computing \& Simulation (project: EP/T001062/1), Quantum Computing via Integrated and Interconnected Implementations (project: EP/Z53318X/1) and EPSRC projects (EP/W032643/1, EP/Y004752/1, and EP/W524323/1 [2928584]), EU Horizon 2020 Marie Sklodowska-Curie Innovation Training Network (project: 956071, `AppQInfo'), National Research Council of Canada (project QSP 062-2), and KIST through the Open Innovation fund and the National Research Foundation of Korea grant funded by the Korean government (MSIT) (No. RS-2024-00413957). Z.L. acknowledges partial funding support from ORCA Computing.
H.Z. and L.Z. acknowledge funding support from National Natural Science
Foundation of China (Grant Nos.~12347104, U24A2017), and the Natural Science Foundation of Jiangsu Province (Grant No.~BK20243060).
\end{acknowledgments}

\bibliography{references}

\section{Methods}
\subsection{Joint overlap estimation scheme}

In our overlap estimation scheme, we interfere the two photonic states on balanced BSs before measuring one output register by PNRDs. The balanced BSs are described by unitary operator $\hat{U}_{\textup{BS}}$ that combine the modes pairwise: 
\begin{eqnarray}
    \hat{U}_{\textup{BS}}^\dagger \hat{a}^{(A)}_{i} \hat{U}_{\textup{BS}}
    &= \frac{1}{\sqrt{2}}\left(\hat{a}^{(A)}_i + \hat{a}^{(B)}_i \right),
    \\
    \hat{U}_{\textup{BS}}^\dagger \hat{a}^{(B)}_{i} \hat{U}_{\textup{BS}} 
    &= \frac{1}{\sqrt{2}}\left(\hat{a}^{(A)}_i - \hat{a}^{(B)}_i \right).
\end{eqnarray}

Taking a partial trace over register $(A)$ after the BS operation, the output state in register $(B)$, denoted $\hat{\rho}^{\textup{(out, B)}}$, has characteristic function, 
\begin{equation}\label{eqn: output characteristic function}
    \chi^{\textup{(out, B)}}(\bm{\alpha}) 
    =
    \chi^{(A)}\left(\frac{\bm{\alpha}}{\sqrt{2}}\right)
    \chi^{(B)} \left(-\frac{\bm{\alpha}}{\sqrt{2}}\right).
\end{equation}
With a multiplicative rescaling, integrating Equation~\ref{eqn: output characteristic function} over phase space recovers the overlap expression in Equation~\ref{eqn: overlap definition}. This in fact equals the Wigner function of the output state at the origin of phase space, which can be measured by a total photon-number parity measurement across the $M$ modes in register $(B)$: 
\begin{eqnarray}
    W^{\textup{(out)}}(\bm{0})
    &=
    \frac{1}{\pi^{2M}}
    \int_{\mathds{C}^M} d^{2M}\bm{\alpha}\, 
    \chi^{\textup{(out)}}(\bm{\alpha})
    \nonumber 
    \\
    &=
    \Tr_{(B)}\left[
    \hat{\rho}^{\textup{(out, B)}}
    \hat{\Pi}
    \right]
\end{eqnarray}
This recovers our overlap measurement scheme described by Equation~\ref{eqn: overlap measurement by parity}. 

Operationally, we can sample $N$ shots and record $\{x_k\}_{k=1}^N$, where $x_k=+1$ for the $k$-th shot if the PNRDs measure an even total photon number and $x_k=-1$ if odd. From the binary outcomes, we build an unbiased estimator for the overlap value by $\widetilde{Y}=\frac{1}{N}\sum_{k=1}^N x_k$, with expectation value $\langle\widetilde{Y}\rangle = \Tr\left[\hat{\rho}^{(A)} \hat{\rho}^{(B)}\right]$ by Equation~\ref{eqn: overlap measurement by parity}. 

\begin{figure*}[!t]
    \centering
    \includegraphics[width=0.98\textwidth]{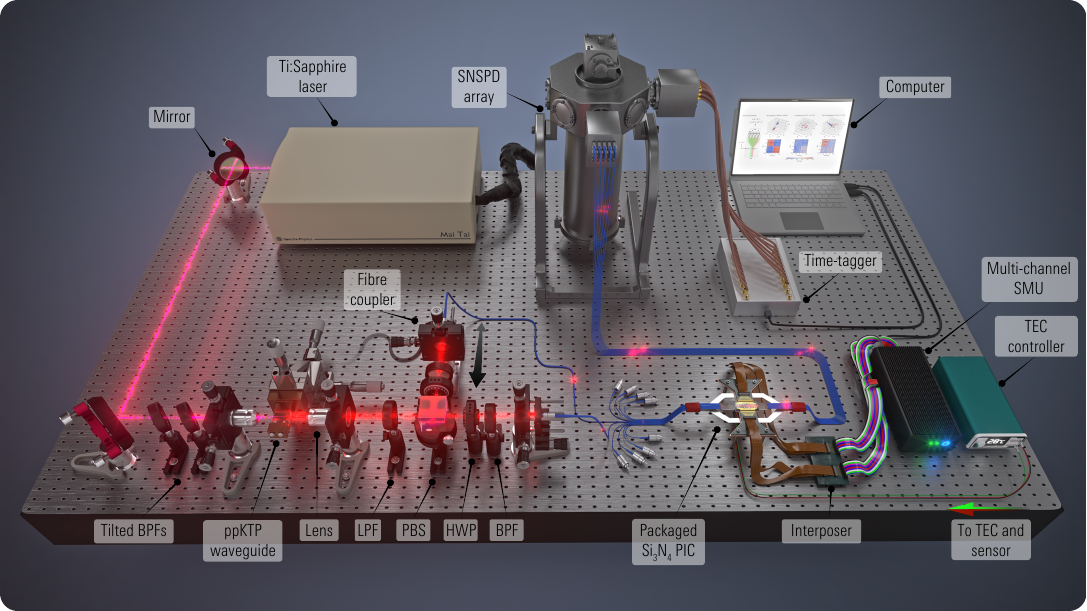}
    \caption{An illustration of the experimental setup. A Ti:Sapphire laser produces 775~nm wavelength, $\sim$1~ps long pulses, at 80~MHz, which are spectrally filtered to 0.9~nm bandwidth (full-width at half-maximum) using a pair of angle-tuned bandpass filters (BPF). The pulses are coupled into a 10~mm long ppKTP waveguide that is phase-matched for type-II SPDC. The residual pump is filtered using a long-pass filter (LPF) and a polarising beamsplitter (PBS) separates signal and idler pulses along separate paths, each containing a HWP and a BPF with a matched bandwidth to eliminate spectral correlations between signal and idler. The photon pair is delivered to the fully programmable $\text{Si}_3\text{N}_4$ PIC via polarisation-maintaining (PM) fibre. The thermo-optic phaseshifters on the PIC are programmed by a multi-channel source-measure unit (SMU), and a thermoelectric cooler (TEC) maintains the temperature of the PIC at $28^{\circ}\text{C}$. The eight output channels of the PIC are coupled to PM fibre, enabling detection by superconducting nanowire single-photon detectors (SNSPDs).
    }
    \label{fig:setup}
\end{figure*}

\subsection{Sample complexity optimality}
The proof for Theorem~\ref{thm: optimality of joint overlap estimation} relies on converting the overlap estimation problem into a decisional quantum state discrimination problem, whose sample complexity can be bounded by applying Helstrom's bound~\cite{helstrom_1969, Anshu-2022-distributed_quantum_inner_product}.

In a quantum state discrimination problem, Alice chooses and prepares one state from $\hat{\sigma}_{-}^{\otimes N}=(\hat{\rho}_{-}\otimes |\phi_0\rangle\langle\phi_0|)^{\otimes N}$ or $\hat{\sigma}_{+}^{\otimes N}=(\hat{\rho}_{+}\otimes |\phi_0\rangle\langle\phi_0|)^{\otimes N}$ with equal prior probability, and sends it to Bob, where $\hat{\rho}_{\pm}=|\psi_{\pm}\rangle\langle\psi_{\pm}|$ are two pure states defined by, 
\begin{equation}
    |\psi_{\pm}\rangle
    =
    \sqrt{\frac{1}{2}\pm\epsilon} |\phi_0\rangle + 
    \sqrt{\frac{1}{2}\mp\epsilon} |\phi_1\rangle. 
\end{equation}
The pure states $|\phi_{0/1}\rangle$ need to be orthogonal---they can be the computational basis states in DV systems or squeezed vaccum and single-photon-subtracted  squeezed vacuum states in CV systems. 

For the quantum state discrimination problem, Bob can distinguish between $\hat{\sigma}_{\pm}^{\otimes N}$ with success probability $P_\text{success}\geq 1-\delta$ if $N$ suffices to estimate $\Tr\left[\hat{\rho}_{-}\otimes |\phi_0\rangle\langle\phi_0|\right]$ to within $\epsilon$ error with $1-\delta$ probability. However, by Helstrom's bound, the success probability is also upper-bounded by the trace distance between $\hat{\sigma}_{\pm}^{\otimes N}$,
\begin{equation}
    P_{\textup{success}}
    \leq 
    \frac{1}{2}
    \left(1 + \frac{1}{2} ||\hat{\sigma}_{-}^{\otimes N}-\hat{\sigma}_{+}^{\otimes N}||_{1}\right)
    \leq 
    \frac{1}{2} + \epsilon \sqrt{N},
\end{equation}
thus giving the lower bound $N=\Omega(\epsilon^{-2}(1-2\delta)^2/4)$ in Theorem~\ref{thm: optimality of joint overlap estimation}, which simplifies to $N=\Omega(\epsilon^{-2})$ if we for example set $\delta=\frac{1}{3}$. For any method that solves the overlap estimation problem, the special states we defined above, $\hat{\sigma}_{\pm}^{\otimes N}$, constitute worst-case instances where no method can guarantee a sample complexity better than $N=O(\epsilon^{-2})$ without violating the Helstrom's bound. The formal proof of Theorem~\ref{thm: optimality of joint overlap estimation} is provided in the Supplementary Information. 

\subsection{Distributed overlap estimation for CV states}

In the main text, we consider a specific distributed overlap estimation method for self-reflective CV states. For these states, independent characterisation is efficient, whereas distributed overlap estimation remains inefficient.

A state is self-reflective if there exists a set of axes in phase space through the origin, about which the reflection of the state is itself. The worst-case complexity of estimating the characteristic function of CV states is exponential in the mode number~\cite{coroi2025exponentialadvantage}. However, for self-reflective states, an efficient $N$-copy measurement exists that can estimate their $\chi(\bm{\alpha}_i)$ up to additive error $\tilde{\epsilon}$ for $L$ points in phase space given only logarithmic sample complexity $N = O\left(\tilde{\epsilon}^{-4}\log(L)\right)$~\cite{Wu-2024-efficient_learning}. 

This provides a route to distributed overlap estimation. Alice and Bob respectively produce characteristic function estimators, $\tilde{\chi}^{(A/B)}(\bm{\alpha}_i)$, with error $\tilde{\epsilon}$ for $L$ points in some sub-space $\mathcal{A}$. Then they estimate the overlap by Monte-Carlo approximation to Equation~\ref{eqn: overlap definition}: $\tilde{Y} = \frac{|\mathcal{A}|}{\pi^M L} \sum_{i=1}^L \tilde{\chi}^{(A)}(\bm{\alpha}_i) \tilde{\chi}^{(B)}(-\bm{\alpha}_i)$. If we bound the energy of the CV state and assume the integral of Equation~\ref{eqn: overlap definition} is negligible outside $\mathcal{A}$, then $\tilde{Y}$ estimates the overlap with error $\epsilon$, if $\tilde{\epsilon} = \frac{\pi^M}{8|\mathcal{A}|}\epsilon$ and $L = \left(\frac{4 \sigma_L |\mathcal{A}|}{\pi^M \epsilon}\right)^2$, where $M$ is the mode number. The Monte-Carlo error is given by $(\sigma_L)^2 = \sum_{i=1}^L \left( f(\bm{\alpha}_i) - \sum_{j=1}^L f(\bm{\alpha}_j)/L \right)/(L-1)$ for $f(\bm{\alpha}_i) = \chi^{(A)}(\bm{\alpha}_i) \chi^{(B)}(-\bm{\alpha}_i)$ and can be estimated using the $\tilde{\chi}^{(A/B)}$ estimators. The resulting sample complexity is $N = O\left( |\mathcal{A}|^4/\epsilon^4 \pi^{4M} \log(|\mathcal{A}|^2\sigma_L^2/\epsilon^2\pi^{2M}) \right)$. 

However, if we adopt a commonly used mode-extensive energy constraint on the CV state~\cite{mele2024learningquantumstatescontinuous,coroi2025exponentialadvantage}, such that the sub-space $\mathcal{A}$ is a $2M$-dimensional hypersphere with radius $|\bm{\alpha}|=\sqrt{\kappa M}$, then it has volume, $|\mathcal{A}| = \frac{(\pi \kappa M)^M}{M!} = O\left(\frac{(\mathrm{e}\pi\kappa)^M}{\sqrt{2\pi M}}\right)$. As a result, the overall sample complexity is still exponential in $M$, as stated informally in Theorem~\ref{theorem: informal distributed} (formal proof in Supplementary Information). 

\subsection{Photon source}
Photon pairs are generated using a type‑II SPDC source based on a ppKTP waveguide (AdvR) pumped by a mode‑locked Ti:Sapphire laser (Spectra-Physics Mai Tai, 80~MHz repetition rate, 775~nm), with the pump spectrum filtered to a $0.9$~nm FWHM using a pair of 3~nm bandpass filters (Thorlabs FBH780-3). The source produces collinear, orthogonally polarised signal and idler photons at 1550~nm, which are separated by a polarising beam splitter and mildly filtered using a 10~nm bandpass filter, before being coupled into polarisation‑maintaining (PM) fibres, and injected into a chip‑based photonic circuit with output modes detected by SNSPDs. At a pump power 65~$\mu$W, the detected signal–idler coincidence rate is 5.5~kHz, while higher‑order photon‑number contributions are suppressed to below $50$ per second. Two‑photon indistinguishability is characterised via on‑chip Hong–Ou–Mandel interference with a variable delay in the signal mode, yielding a visibility of $95.4\%$.

\subsection{Photonic chip calibration}
The silicon nitride chip used in this work, fabricated by LioniX, comprises a 10$\times$10 array of thermo‑optically tunable phaseshifters and is packaged with PM fibre arrays for both input and output coupling (PHIX Photonics Assembly BV). 
The chip implements a 10-mode universal linear-optical interferometer.
Before performing QML tasks, each phaseshifter is carefully calibrated, with directional coupler imperfections and thermal crosstalk taken into account. We further develop a crosstalk‑mitigation method that improves the average fidelity of the output coincidence distribution to $0.980\pm 0.006$. In addition, a digital simulator of the chip is designed to quantify the detrimental effects of thermal crosstalk and to validate the effectiveness of our calibration. Further details of the calibration methods and the simulator are provided in the Supplementary Information. 
\par

\subsection{Support vector machine}

Not all datasets are linearly separable in their original space. In this scenario, the trick is to find a non-linear kernel function that maps the data into a high-dimensional space where the task can be reduced to linear classification. For data vectors with dimension $d$, denoted $\bm{x}^{(i)} \in \mathds{R}^d$, one can encode the data into some quantum state $|\psi(\bm{x}^{(i)})\rangle$, which constitutes a quantum feature map into Hilbert space $\mathcal{H}$. 

For each dataset classified by SVM in main text, we generate a training subset of $100$ binary-labelled data instances, $\{|\psi(\bm{x}^{(i)})\rangle, y^{(i)}\}_{i=1}^{100}$. By evaluating the overlap between each pair of data instances, we construct a $100\times 100$ overlap matrix, $\bm{K}$, whose elements are given by $K_{ij} = |\langle\psi(\bm{x}^{(i)})|\psi(\bm{x}^{(j)})\rangle|^2$. We aim to find a vectorised classifier weight, $\bm{\beta}=(\beta^{(1)}, \hdots, \beta^{(m)})^T$, which maximises the hyperplane separation between the two different classes. We do so by solving the following constrained optimisation problem: 
\begin{equation}
\begin{split}
    &\min_{\bm{\beta}} 
    \frac{1}{2} \bm{\beta}^T \left(\bm{y}^T \bm{K} \bm{y}\right) \bm{\beta}
    - \sum_{i} \beta_i ,
    \\
    \textup{subject to }
    & 0 \leq \beta^{(i)} \leq C,
    \\
    & \bm{y}^T\bm{\beta} = 0,
\end{split}
\end{equation}
where vector $\bm{y} = (y^{(1)}, \hdots, y
^{(m)})^T$ is the vectorised training labels, and $C$ is a slack constant for soft-margin optimisation, which we set to $C=0.8$. Then we define a bias correction term, $b$, by 
\begin{equation}
    b = \frac{1}{100} 
    \sum_{j=1}^{100} 
    \left(y^{(j)} - \sum_{i=1} \beta^{(i)} y^{(j)} K_{ji}\right).
\end{equation}

As a result, for any new unseen data, $\bm{x}$, encoded in qudit $|\psi(\bm{x})\rangle$, we compute its overlap with each training qudit and predict its label by, 
\begin{equation}
    \hat{y}(\bm{x}) 
    = 
    \sign
    \left(
    \sum_{i=1}^{100} \beta^{(i)} y^{(i)} 
    \left|\langle\psi(\bm{x}^{(i)})|\psi(\bm{x})\rangle\right|^2
    + b
    \right).
\end{equation}
In practice, however, the classifier weight $\beta^{(i)}$ is typically non-zero for a subset of the training data, which is referred to as the `support vectors' and to which we can limit our overlap evaluations. 

For the three datasets shown in Fig.~\ref{fig: svm results} (linearly separable, concentric spheres and overlapping blobs), we generate $200$ phase vectors according to the dataset structure, which are randomly divided into training and test subsets. The amplitudes of the qudit, $A_j$ in Equation~\ref{eqn: encoded qudit state}, are chosen to optimise the classification accuracy across the three datasets in numerical simulations. Details of the SVM implementation are provided in the Supplementary Information.

\subsection{Simultaneous perturbation stochastic approximation}
For online learning of quantum data, we implement the SPSA algorithm proposed and detailed in Ref.~\cite{spall1998-overview-spsa, spall1998-spsa-implementations}. The routine works by randomly perturbing the reference state at each iteration, and use the overlap evaluation at both directions of the perturbation for gradient estimation. Then the state is updated along the descent direction, before being perturbed at the next iteration.

In SPSA, we seek to minimise a cost function, which in this implementation is defined as one minus the overlap between the target and reference qudit states: 
\begin{equation}
    c(\bm{\theta}) = 1 - \left|\langle \psi(\bm{x}^{(t)})|\psi(\bm{\theta})\rangle\right|^2 .
\end{equation}

In the $k$-th iteration, where the reference qudit has phase, $\bm{\theta}^{(k)}$, we evaluate the cost function at $c\left(\bm{\theta}^{(k)} \pm \bm{\Delta}^{(k)} t^{(k)} \right)$, where $\bm{\Delta}^{(k)} \in \{\pm 1 \}^{\otimes 3}$ is a random perturbation vector. This allows us to approximate the local gradient (vectorised) to be, 
\begin{equation}
    \bm{g}(\bm{\theta}^{(k)}) = 
    \frac{
    c\left(\bm{\theta}^{(k)} + \bm{\Delta}^{(k)} t^{(k)}\right) 
    - c\left(\bm{\theta}^{(k)} - \bm{\Delta}^{(k)} t^{(k)}\right)
    }{
    2 t^{(k)} \bm{\Delta}^{(k)}
    }.
\end{equation}
Then, we update the phase vector for the $(k+1)$-th iteration to be, $\bm{\theta}^{(k+1)} = \bm{\theta}^{(k)} - a^{(k)} \bm{g}(\bm{\theta}^{(k)})$. 

The gain coefficient and perturbation coefficient decay with the number of iterations as: 
\begin{align}
    a^{(k)} &= \frac{a}{(A+k+1)^{\alpha}}, \\
    t^{(k)} &= \frac{t}{(k+1)^{\gamma}}.
\end{align}
For constants, we choose $\alpha=0.602$, $\gamma=0.101$, $A=10$, $a=1.6$ for $500$ iterations per optimisation. At the start of each optimisation, we evaluate the initial cost function at the $0$-th iteration five times and set coefficient $t$ to be twice the standard deviation of the function value. 

We repeat the optimisation for ten target states with randomly selected phases, and for different number of state copies per overlap evaluation, as explained in the main text. The results are summarised in Fig.~\ref{fig: spsa results} with further details provided in the Supplementary Information. 

\subsection{Data collection}
The overlap estimation in our experiment consists of measuring eight optical modes with SNSPDs (Photon Spot), which are non-photon-number-resolving detectors but suffice for measuring the overlap between two single-photon qudit states. The two output registers, $A/B$, labelled in Equation~\ref{eqn: overlap measurement by parity}, correspond to the odd-number-labelled modes $\{1,3,5,7\}$ and even-number-labelled modes $\{2,4,6,8\}$, respectively. Consequently, an odd-parity total-photon-number measurement corresponds to a coincidence detection between two modes $k$ and $l$ where $k+l$ is an odd number. 

In measurement, we post-select outputs on two-photon coincidence events. Due to the suppression of higher-order photon-number components in the photon source, post-selection eliminates the effect of photon loss without introducing unwanted terms to the computation. For a total of $N$ two-photon detection events, comprising $N_{\text{odd}}$ odd-parity events, the overlap estimator is constructed by, 
\begin{equation}\label{eqn: practical formula for overlap estimation}
    \tilde{Y}= 
    1 - 2(1-R) \frac{N_\text{odd}}{N},
\end{equation}
where $R=\sum_{i=0}^3 A_i^4 < 1$. Here, $A_i$ is the encoded amplitude of the qudit states (Equation~\ref{eqn: encoded qudit state}) and are experimentally validated. We treat these amplitudes as known quantities, as data is encoded exclusively in the phase-degrees of freedom. Physically, the term $R$ represents the bunching probability, which is the probability of two photons arriving in the same output mode. These instances contribute to even-parity events but are discarded by our post-selection process. This term can be removed in an improved future implementation with intrinsic PNRDs. 
Additionally, Equation~\ref{eqn: practical formula for overlap estimation} can be rescaled by the Hong-Ou-Mandel interference visibility, though this rescaling was found to have negligible impact on the final accuracy of the QML tasks demonstrated in this work. 

For the quantum data classification experiment, we limit each overlap estimation to $N=1000$ detected coincidences. During data collection, we continuously monitor the coincidence rate and adjust the collection window to measure $15000$ coincidence events per measurement. In post-processing, we randomly sample $1000$ events to perform the overlap estimation, and bootstrap the process $1000$ times. This provides an error estimate for each overlap value, which is then propagated using Monte-Carlo simulations to estimate the error of the final classification accuracy (shown in Fig.~\ref{fig: svm results}). 

In the quantum data online learning experiment, the adaptive nature of the measurement process precludes the use of downsampling and bootstrapping. Instead, we lower the coincidence rate such that each collection window directly produces the target total coincidence number per overlap evaluation. This collection window is continuously monitored and adjusted to accommodate count-rate fluctuations in the photon source. The experiment was repeated for total coincidence numbers $N=10^2, 10^3$ and $10^4$, as presented in Fig.~\ref{fig: spsa results}.

\medskip

\clearpage
\widetext
    
\begin{center}
\textbf{\large Supplementary Information: Machine learning of quantum data using optimal similarity measurements}
\end{center}
\setcounter{equation}{0}
\setcounter{figure}{0}
\setcounter{section}{0}
\setcounter{theorem}{0}
\makeatletter

\renewcommand{\theequation}{S\arabic{equation}}
\renewcommand{\thesection}{S\Roman{section}}
\renewcommand{\thefigure}{S\arabic{figure}}
\renewcommand{\thetheorem}{S\arabic{theorem}}

\section{Photonic states}

\subsection{Preliminary}
In this section, we give an overview of bosonic systems and derive some useful identities. More details can be found in Refs.~\cite{barnett_radmore2002methods, serafini2023quantum, schleichQuantumOpticsPhase2001}.

In the rest of the document, without further explanation, we consider $M$-mode bosonic states described by their density operators $\hat{\rho}$. For discrete-variable (DV) states, the density operator $\hat{\rho}$ acts on a $d$-dimensional Hilbert space. For continuous-variable (CV) states, the Hilbert space dimension is theoretically infinite. 

An $M$-mode bosonic system is defined by its bosonic annihilation and creation operators. For the $k$-th mode, the creation and annihilation operators are denoted as $\hat{a}_k^{\dagger}, \hat{a}_k$ respectively, with commutation relations given by
\begin{equation}
    [\hat{a}_k, \hat{a}_{k'}^\dagger] = \delta_{k,k'}, 
    \quad 
    [\hat{a}_k, \hat{a}_{k'}] = 
    [\hat{a}_k^{\dagger}, \hat{a}_{k'}^{\dagger}]=0.
\end{equation}
The photon number operator of the $k$-th mode is denoted,
\begin{equation}
    \hat{n}_k \coloneqq \hat{a}_k^\dagger \hat{a}_k.
\end{equation}
And we also define a total photon number parity operator: 
\begin{equation}\label{supp eqn: total photon number parity operator}
    \hat{\Pi} \coloneqq (-1)^{\sum_{k=1}^M \hat{n}_k}. 
\end{equation}

The $M$-mode displacement operator is defined as
\begin{equation}\label{supp eqn: multimode displacement operator}
    \hat{D}(\bm{\alpha}) 
    \coloneqq e^{
    \sum_{k=1}^M 
    \left(
    \alpha_k\hat{a}_k^{\dagger} - 
    \alpha_k^*\hat{a}_k
    \right)
    }
\end{equation}
for a complex displacement vector, $\bm{\alpha}=(\alpha_1, \hdots, \alpha_M) \in\mathds{C}^M$, by which the state is displaced in phase space. 
Acting the displacement operator onto the vacuum state generates an $M$-mode coherent state, $|\bm{\alpha}\rangle=\hat{D}(\bm{\alpha})|0\rangle$. Coherent states are not orthogonal but form an overcomplete basis of the Hilbert space, allowing the identity operator to be decomposed as 
\begin{equation}
    \hat{\bm{I}} = 
    \frac{1}{\pi^M}
    \int_{\mathds{C}^M} d^{2M}\bm{\alpha}\, |\bm{\alpha}\rangle\langle\bm{\alpha}|,
\end{equation}
where the integral over the multi-mode phase space is $\frac{1}{\pi^M} \int_{\mathds{C}^M} d^{2M} \bm{\alpha}= \left(\frac{1}{\pi} \int_{\mathds{C}} d^{2} \alpha_1\right) \hdots \left(\frac{1}{\pi} \int_{\mathds{C}} d^{2} \alpha_M\right)$. 

The trace of any trace class operator, $\hat{O}$ can be expressed as an integral over the coherent state basis: 
\begin{equation}\label{supp eqn: trace as integral over coherent state basis}
    \Tr[\hat{O}] = 
    \frac{1}{\pi^M} \int_{\mathds{C}^M} d^{2M}\bm{\alpha}\, \langle \bm{\alpha}|\hat{O}| \bm{\alpha}\rangle.
\end{equation}

The integral of the displacement operator over phase space is~\cite{barnett_radmore2002methods},
\begin{equation}\label{supp eqn: integral of multimode displacement operator}
\begin{split}
    \frac{1}{(2\pi)^M}
    \int_{\mathds{C}^M} d^{2M}\bm{\alpha}\, 
    \hat{D}(\bm{\alpha}) 
    &=
    \prod_{k=1}^M
    \left(
    \frac{1}{2\pi} \int_{\mathds{C}} d^2\alpha_k\, 
    \hat{D}_k(\alpha_k) 
    \right)
    \\
    &=
    \prod_{k=1}^M
    \left(
    \sum_{m=0}^\infty 
    \frac{(-2)^m}{m!} (\hat{a}_k^\dagger)^m \hat{a}_k^m
    \right)
    \\
    &=
    \prod_{k=1}^M
    \left( (-1)^{\hat{a}_k^\dagger \hat{a}_k} \right)
    =
    (-1)^{\sum_{k=1}^M \hat{n}_k} \equiv \hat{\Pi}
\end{split}
\end{equation}
which is the total photon number parity operator in Equation~\ref{supp eqn: total photon number parity operator}.

Using Equation~\ref{supp eqn: trace as integral over coherent state basis}, the trace of the displacement operator is
\begin{equation}\label{supp eqn: trace of displacement operator}
\begin{split}
    \Tr\left[\hat{D}(\bm{\alpha})\right]
    &=
    \frac{1}{\pi^M}
    \int_{\mathds{C}^{M}} d^{2M} \bm{\beta}\,
    \langle\bm{\beta}|\hat{D}(\bm{\alpha})|\bm{\beta}\rangle
    \\
    &=
    \frac{1}{\pi^M}
    \int_{\mathds{C}^{M}} d^{2M} \bm{\beta}\,
    e^{
    \bm{\beta}^{\dagger}\bm{\alpha}-\bm{\alpha}^{\dagger}\bm{\beta}}
    e^{-\frac{1}{2}\bm{\alpha}^\dagger \bm{\alpha}}
    =
    \pi^M \delta^{(2M)}(\bm{\alpha}). 
\end{split}
\end{equation}

\subsection{Description of photonic states}
Leveraging the completeness relation of the coherent states, any $M$-mode state with density operator $\hat{\rho}$ can be decomposed as, 
\begin{equation}\label{supp eqn: photonic state in characteristic function}
    \hat{\rho}
    =
    \frac{1}{\pi^M}
    \int_{\mathds{C}^M}
    d^{2M}\bm{\alpha}\,
    \chi(\bm{\alpha}) \hat{D}(-\bm{\alpha}),
\end{equation}
for characteristic function 
\begin{equation}\label{supp eqn: characteristic function definition}
    \chi(\bm{\alpha}) \coloneqq \Tr\left[\hat{D}(\bm{\alpha})\hat{\rho}\right].
\end{equation}
A photonic state can be fully described by its characteristic function.  

Alternatively, we can take the complex Fourier transform of Equation~\ref{supp eqn: characteristic function definition} and derive the Wigner function description of the state, which is
\begin{equation}\label{supp eqn: wigner function definition}
    W(\bm{\alpha})
    =
    \frac{1}{\pi^{2M}}
    \int_{\mathds{C}^M} d^{2M}\bm{\beta}\,
    e^{\bm{\beta}^\dagger \bm{\alpha} - \bm{\alpha}^\dagger \bm{\beta}}
    \chi(\bm{\beta}). 
\end{equation}
Specifically, we have,
\begin{equation}
\begin{split}
    e^{\bm{\beta}^\dagger\bm{\alpha}-\bm{\alpha}^\dagger\bm{\beta}}
    \hat{D}(\bm{\beta})
    &=
    e^{
    \sum_{k=1}^M 
    \left[
    \beta_k(\hat{a}_k^{\dagger}-\alpha_k^*) - 
    \beta_k^*(\hat{a}_k-\alpha_k)
    \right]
    }
    \\
    &=
    e^{
    \sum_{k=1}^M 
    \hat{D}_k^\dagger(-\alpha_k)
    \left(
    \beta_k \hat{a}_k^{\dagger}
    -\beta_k^* \hat{a}_k 
    \right)
    \hat{D}_k(-\alpha_k)
    }
    \\
    &=
    \hat{D}^\dagger(-\bm{\alpha})
    \hat{D}(\bm{\beta}) \hat{D}(-\bm{\alpha}),
\end{split}
\end{equation}
which, after substituting into Equation~\ref{supp eqn: wigner function definition} and using the definition of Equation~\ref{supp eqn: characteristic function definition}, allows us to write the Wigner function explicitly as, 
\begin{equation}\label{supp eqn: explicit form of Wigner function in parity operator}
\begin{split}
    W(\bm{\alpha})
    &=
    \frac{1}{\pi^{2M}}
    \int_{\mathds{C}^M} d^{2M}\bm{\beta}\,
    \Tr\left[
    e^{\bm{\beta}^\dagger \bm{\alpha} - \bm{\alpha}^\dagger \bm{\beta}}
    \hat{D}(\bm{\beta})
    \hat{\rho}
    \right]
    \\
    &=
    \left(\frac{2}{\pi}\right)^M 
    \Tr\left[
    \hat{D}(\bm{\alpha})
    \hat{\rho} \hat{D}^\dagger(\bm{\alpha})
    \frac{1}{(2\pi)^M}
    \int_{\mathds{C}^M} d^{2M}\bm{\beta}\,\hat{D}(\bm{\beta})
    \right]
    \\
    &=
    \left(\frac{2}{\pi}\right)^M 
    \Tr\left[
    \hat{D}(\bm{\alpha})
    \hat{\rho} \hat{D}^\dagger(\bm{\alpha})
    \hat{\Pi}
    \right], 
\end{split}
\end{equation}
where we used Equation~\ref{supp eqn: integral of multimode displacement operator} to evaluate the integral of the displacement operator. 
Operationally, this allows us to directly measure the Wigner function of a quantum state by first displacing it in phase space, and then measuring its total photon number parity. 
In particular, at origin of phase space, the displacement is zero and we can directly measure the Wigner function by PNRDs without displacement: 
\begin{equation}\label{supp eqn: wigner function is photon number parity measurement}
    W(\bm{0})=\left(\frac{2}{\pi}\right)^{M}
    \Tr\left[\hat{\rho}\hat{\Pi}\right]. 
\end{equation}

\subsection{Overlap between two states}
For two $M$-mode states, $\hat{\rho}^{(A)}$ and $\hat{\rho}^{(B)}$, their overlap can be defined as $\Tr\left[\hat{\rho}^{(A)}\hat{\rho}^{(B)}\right]$. For two pure states, $\hat{\rho}^{(A)}=|\psi\rangle\langle\psi|, \hat{\rho}^{(B)}=|\phi\rangle\langle\phi|$, this definition equals $|\langle\psi|\phi\rangle|^2$. Expressed in terms of their characteristic functions, $\chi^{(A)}(\bm{\alpha}), \chi^{(B)}(\bm{\alpha})$, this is 
\begin{equation}\label{supp eqn: overlap in characteristic function form}
\begin{split}
    \Tr\left[\hat{\rho}^{(A)}\hat{\rho}^{(B)}\right]
    &=
    \frac{1}{\pi^{2M}}
    \int_{\mathds{C}^M}
    d^{2M}\bm{\alpha}\,
    \int_{\mathds{C}^M}
    d^{2M}\bm{\beta}\,
    \chi^{(A)}(\bm{\alpha}) 
    \chi^{(B)}(\bm{\beta})
    \Tr\left[
    \hat{D}(-\bm{\alpha}) 
    \hat{D}(-\bm{\beta})
    \right]
    \\
    &=
    \frac{1}{\pi^{2M}}
    \int_{\mathds{C}^{2M}}
    d^{2M}\bm{\alpha}
    d^{2M}\bm{\beta}\,
    \chi^{(A)}(\bm{\alpha}) 
    \chi^{(B)}(\bm{\beta}) 
    \Tr\left[
    \hat{D}(-\bm{\alpha}-\bm{\beta})
    \right]
    e^{\frac{1}{2}(\bm{\beta}^{\dagger}\bm{\alpha}
    -\bm{\alpha}^{\dagger}\bm{\beta})}
    \\
    &=
    \frac{1}{\pi^M}
    \int_{\mathds{C}^{2M}}
    d^{2M}\bm{\alpha}\,
    \chi^{(A)}(\bm{\alpha}) 
    \chi^{(B)}(-\bm{\alpha}), 
\end{split}
\end{equation}
where in the second line we used the combination of two displacement operators, $\hat{D}(\bm{\alpha})\hat{D}(\bm{\beta}) = e^{\frac{1}{2}(\bm{\beta}^{\dagger}\bm{\alpha}-\bm{\alpha}^{\dagger}\bm{\beta})}\hat{D}(\bm{\alpha}+\bm{\beta})$.

\section{Joint overlap estimation}
In this section, we explicitly derive the joint overlap estimation scheme given in the main text, including its sample complexity and the optimality of the sample complexity. 

\subsection{Derivation}
We consider two $M$-mode states, $\hat{\rho}^{(A)}$ and $\hat{\rho}^{(B)}$, described by characteristic functions $\chi^{(A/B)}(\bm{\alpha})$. Each state occupies half of a total $2M$-mode system. If we divide the $2M$ modes into two registers, we will use superscripts $(A/B)$ to label the register where the input state originally sits in, and subscripts $i=1, \hdots, M$ to label the modes within each register. 

We interfere the two $M$-mode states pairwise on balanced beamsplitters (BSs) described by unitary operator, $\hat{B}$, with action
\begin{equation}
    \hat{B}^\dagger \hat{a}_i^{(A)} \hat{B}
    = 
    \frac{1}{\sqrt{2}}
    \left(
    \hat{a}_i^{(A)} + \hat{a}_i^{(B)}
    \right)
    ,\quad 
    \hat{B}^\dagger \hat{a}_i^{(B)} \hat{B}
    = 
    \frac{1}{\sqrt{2}}
    \left(
    \hat{a}_i^{(A)} - \hat{a}_i^{(B)}
    \right).
\end{equation}
For example, a $2M$-mode displacement described by $\hat{D}^{(A)}(\bm{\alpha})\otimes \hat{D}^{(B)}(\bm{\beta})$ for complex vectors $\bm{\alpha}, \bm{\beta} \in \mathds{C}^M$ undergoes the following transformation under the beamsplitters: 
\begin{equation}\label{supp eqn: displacement after beamsplitters}
    \hat{B}^\dagger 
    \left(
    \hat{D}^{(A)}(\bm{\alpha})\otimes \hat{D}^{(B)}(\bm{\beta})
    \right)
    \hat{B}
    =
    \hat{D}^{(A)}\left(\frac{\bm{\alpha}+\bm{\beta}}{\sqrt{2}}\right)
    \otimes 
    \hat{D}^{(B)}\left(\frac{\bm{\alpha}-\bm{\beta}}{\sqrt{2}}\right).
\end{equation}

Expanding the states, $\hat{\rho}^{(A/B)}$, in terms of Equation~\ref{supp eqn: photonic state in characteristic function} and using Equation~\ref{supp eqn: displacement after beamsplitters}, we derive, 
\begin{equation}
    \hat{B}^\dagger 
    \left(\hat{\rho}^{(A)} \otimes \hat{\rho}^{(B)} \right)
    \hat{B}
    =
    \frac{1}{\pi^{2M}}
    \int_{\mathds{C}^M} d^{2M}\bm{\alpha}\, 
    \int_{\mathds{C}^M} d^{2M}\bm{\beta}\, 
    \chi^{(A)}(\bm{\alpha}) \chi^{(B)}(\bm{\beta})
    \hat{D}^{(A)}\left(-\frac{\bm{\alpha}+\bm{\beta}}{\sqrt{2}}\right)
    \hat{D}^{(B)}\left(-\frac{\bm{\alpha}-\bm{\beta}}{\sqrt{2}}\right).
\end{equation}
Taking a partial trace over the register $(A)$, and using Equation~\ref{supp eqn: trace of displacement operator}, the output state in the register $(B)$ is 
\begin{equation}
    \hat{\rho}^{\textup{(out)}}
    =
    \Tr_{(A)} 
    \left[\hat{B}^\dagger
    \left(\hat{\rho}^{(A)} \otimes \hat{\rho}^{(B)} \right) \hat{B}
    \right]
    =
    \frac{1}{\pi^M}
    \int_{\mathds{C}^M} d^{2M}\bm{\gamma}\, 
    \chi^{(A)}\left(\frac{\bm{\gamma}}{\sqrt{2}}\right)
    \chi^{(B)}\left(-\frac{\bm{\gamma}}{\sqrt{2}}\right)
    \hat{D}^{(B)}(-\bm{\gamma}), 
\end{equation}
which means it has characteristic function $\chi^{\textup{(out)}}(\bm{\gamma}) = \chi^{(A)}\left(\frac{\bm{\gamma}}{\sqrt{2}}\right) \chi^{(B)}\left(-\frac{\bm{\gamma}}{\sqrt{2}}\right)$. 

Using Equation~\ref{supp eqn: wigner function definition}, the Wigner function of the output state at origin of phase space is given by 
\begin{equation}\label{supp eqn: wigner function is overlap}
\begin{split}
    W^{\textup{(out)}}(\bm{0})
    &=
    \frac{1}{\pi^{2M}}
    \int_{\mathds{C}^M} d^{2M}\bm{\gamma}\, 
    \chi^{\textup{(out)}}(\bm{\gamma})
    \\
    &=
    \frac{2^M}{\pi^{2M}}
    \int_{\mathds{C}^M} d^{2M}\bm{\alpha}\,
    \chi^{(A)}(\bm{\alpha}) \chi^{(B)}(-\bm{\alpha})
    \\
    &=
    \left(\frac{2}{\pi}\right)^M 
    \Tr\left[\hat{\rho}^{(A)}\hat{\rho}^{(B)}\right].
\end{split}
\end{equation}
In the second line, we made a substitution of integration variable, $\bm{\gamma}=\sqrt{2}\bm{\alpha}$, and in the third line we used the definition of overlap in Equation~\ref{supp eqn: overlap in characteristic function form}. 

Finally, we can equate Equation~\ref{supp eqn: wigner function is overlap} to Equation~\ref{supp eqn: wigner function is photon number parity measurement}. This gives, 
\begin{equation}\label{supp eqn: joint overlap estimation method}
    \Tr\left[\hat{\rho}^{(A)}\hat{\rho}^{(B)}\right]
    =
    \Tr_{(B)}\biggr[
    \Tr_{(A)} 
    \left[\hat{B}^\dagger
    \left(\hat{\rho}^{(A)} \otimes \hat{\rho}^{(B)} \right) \hat{B}
    \right]
    \hat{\Pi}
    \biggr].
\end{equation}
Operationally, Equation~\ref{supp eqn: joint overlap estimation method} means we can estimate the overlap between two photonic states by first interfering them on balanced BSs before measuring half of the output register on PNRDs for their photon-number parity. 

\subsection{Sample complexity}
The parity measurement in Equation~\ref{supp eqn: joint overlap estimation method} is a binary outcome measurement. Since the sample variance is a fixed quantity for this type of measurements, derivation of their sample complexity commonly uses the Hoeffding's inequality: Let $x_1, \hdots, x_N$ be independent random variables such that $a_k \leq x_k \leq b_k$ for all $k\in[1, N]$. For the sum of these random variables, $S_N=\sum_{k=1}^N x_k$, Hoeffding's inequality states that for all $t > 0$, we have~\cite{Hoeffding1963} 
\begin{equation}\label{supp eqn: hoeffding}
    \Pr\left(
    |S_N - \langle S_N\rangle | \geq t
    \right)
    \leq 
    2 \exp\left(
    - \frac{2t^2}{\sum_{k=1}^N (b_k-a_k)^2}
    \right),
\end{equation}
where $\langle S_N\rangle$ is the expectation value of $S_N$. 

The sample complexity of our joint overlap estimation scheme is stated in Theorem 1 of the main text, which we restate below and give a full proof. 
\begin{theorem}
    \textbf{\textup{(Theorem 1 in main text.)}}
    Given $N$ copies of $\hat{\rho}^{(A)} \otimes \hat{\rho}^{(B)}$ and some constants $\epsilon,\delta\in(0,\frac{1}{2})$, by use of balanced BSs and PNRDs, the overlap $\Tr\left[\hat{\rho}^{(A)} \hat{\rho}^{(B)}\right]$ can be estimated to within additive error $\epsilon$ with success probability at least $1-\delta$ for sample complexity $N=2\epsilon^{-2}\log(2\delta^{-1})$. 
\end{theorem}
\begin{proof}
The overlap can be estimated by the parity measurement in Equation~\ref{supp eqn: joint overlap estimation method}, which uses only balanced BSs and PNRDs. Suppose we perform the experiment for $N$ shots, and the result of each shot is recorded as $x_k$ for $k\in[1,N]$, where $x_k=1$ if an even photon number is measured and $x_k=-1$ if odd. They satisfy $-1\leq x_k \leq 1$, such that $a_k=-1, b_k=1$ for $k\in[1, N]$ in Equation~\ref{supp eqn: hoeffding}. Then the estimator, $\tilde{Y}=\frac{1}{N}\sum_{k=1}^N x_k$, is an unbiased estimator for the overlap, $\langle \tilde{Y} \rangle = \langle \hat{\Pi} \rangle = \Tr\left[\hat{\rho}^{(A)}\hat{\rho}^{(B)}\right]$. If we want to estimate the overlap within an additive error $\epsilon>0$ with probability at least $1-\delta$, this can be written as 
\begin{equation}
    \Pr\left(
    |\tilde{Y} - \langle \tilde{Y} \rangle| \geq \epsilon
    \right)
    =
    \Pr\left(
    \left|
    \sum_{k=1}^N x_k - \left\langle \sum_{k=1}^N x_k \right\rangle
    \right| \geq N \epsilon
    \right)
    \leq \delta. 
\end{equation}

Using Hoeffding's inequality, we have $\delta \geq 2\mathrm{e}^{-\frac{N\epsilon^2}{2}}$, which means the sample complexity is 
\begin{equation}
    N = 2\epsilon^{-2} \log(2\delta^{-1}), 
\end{equation}
thus completing the proof.
\end{proof}

If we give the success probability a number, e.g. $1-\delta=\frac{2}{3}$, then this is $N=2\log(6)\epsilon^{-2}\approx 3.58 \epsilon^{-2}$. Adopting the big-$O$ notation for asymptotic scaling, the sample complexity simplifies to $N=O(\epsilon^{-2})$.

\subsection{Optimality}
The $N=O(\epsilon^{-2})$ sample complexity is, in fact, optimal up to a constant. A proof is given by Ref.~\cite{Anshu-2022-distributed_quantum_inner_product, Wang-2024-optimal_trace_distance} for DV systems, but is easily extended to arbitrary photonic states. For completeness, we detail the proof in this section, after formally restating the result from Theorem 2 in the main text. The proof works by converting the overlap estimation problem into a decisional problem of quantum state discrimination, whose success probability is bounded by Helstrom's bound~\cite{helstrom_1969}. 

\begin{theorem}\label{supp theorem: overlap estimation lower bound}
    \textbf{\textup{(Theorem 2 of main text.)}}
    Given the $N$-copy state, $\left(\rho^{(A)} \otimes \rho^{(B)}\right)^{\otimes N}$, and some constants $\epsilon, \delta\in(0,\frac{1}{2})$, if an algorithm estimates $\Tr\left[\rho^{(A)} \rho^{(B)}\right]$ to within additive error $\epsilon$ with success probability at least $1-\delta$, then $N=\Omega\left(\epsilon^{-2}\left(\frac{1}{2}-\delta\right)^2\right)$. 
\end{theorem} 

\begin{proof}
We consider two orthogonal basis states $|\phi_0\rangle, |\phi_1\rangle$, such that $\langle \phi_1 | \phi_0\rangle =0$. 
For some $\epsilon>0$, we define two pure states, $\hat{\rho}_{\pm}$, that can be decomposed in terms of $|\phi_0\rangle, |\phi_1\rangle$: 
\begin{align}
    \hat{\rho}_{-}&=|\psi_{-}\rangle\langle\psi_{-}|, \quad \textup{for} \quad
    |\psi_{-}\rangle =
    \sqrt{\frac{1}{2}-\epsilon} |\phi_0\rangle + 
    \sqrt{\frac{1}{2}+\epsilon} |\phi_1\rangle,
    \\
    \hat{\rho}_{+}&=|\psi_{+}\rangle\langle\psi_{+}|, \quad \textup{for} \quad
    |\psi_{+}\rangle =
    \sqrt{\frac{1}{2}+\epsilon} |\phi_0\rangle + 
    \sqrt{\frac{1}{2}-\epsilon} |\phi_1\rangle,
\end{align}
which satisfy
\begin{equation}
    \Tr\left[\hat{\rho}_{-} |\phi_0\rangle\langle\phi_0|\right] = \frac{1}{2}-\epsilon, \quad
    \Tr\left[\hat{\rho}_{+} |\phi_0\rangle\langle\phi_0|\right] = \frac{1}{2}+\epsilon. 
\end{equation}

In a quantum state discrimination problem, Alice picks and prepares one state from $\hat{\sigma}_{-}^{\otimes N}=(\hat{\rho}_{-}\otimes |\phi_0\rangle\langle\phi_0|)^{\otimes N}$ or $\hat{\sigma}_{+}^{\otimes N}=(\hat{\rho}_{+}\otimes |\phi_0\rangle\langle\phi_0|)^{\otimes N}$ with equal prior probability, and sends it to Bob. If an algorithm described in Theorem~\ref{supp theorem: overlap estimation lower bound} exists, then Bob can discriminate between the two non-orthogonal states, $\hat{\sigma}_{\pm}^{\otimes N}$, with at least $1-\delta$ probability. 

Therefore, we convert the question of lower bounding the overlap estimation problem's sample complexity to upper bounding a quantum state discrimination problem's success probability. The latter question can be answered by the Helstrom's bound~\cite{helstrom_1969} (or~\cite{Wilde_2013} for a textbook description):
\begin{equation}\label{supp eqn: Helstrom bound}
    1-\delta \leq 
    P_{\textup{success}}
    \leq 
    \frac{1}{2}
    \left(1 + \frac{1}{2} ||\hat{\sigma}_{-}^{\otimes N}-\hat{\sigma}_{+}^{\otimes N}||_{1}\right),
\end{equation}
for trace distance, 
\begin{equation}\label{supp eqn: trace distance derivation}
\begin{split}
    ||\hat{\sigma}_{-}^{\otimes N}-\hat{\sigma}_{+}^{\otimes N}||_{1}
    &=
    \Tr\left[\sqrt{
    \left(\hat{\sigma}_{-}^{\otimes N}-\hat{\sigma}_{+}^{\otimes N}\right)^\dagger 
    \left(\hat{\sigma}_{-}^{\otimes N}-\hat{\sigma}_{+}^{\otimes N}\right)
    }\right]
    \\
    &\leq 
    2 \sqrt{1 - \Tr\left[\hat{\sigma}_{-}^{\otimes N}\hat{\sigma}_{+}^{\otimes N}\right]}
    \\
    &=
    2 \sqrt{1 - |\langle\psi_{+}|\psi_{-}\rangle|^{2N}}
    = 
    2 \sqrt{1 - (1-4\epsilon^2)^N}
    \leq 
    4 \epsilon \sqrt{N},
\end{split}
\end{equation}
where in the second line we used the relationship between trace distance and trace overlap~\cite{Wilde_2013}, and in the last line we used the fact that $(1-4\epsilon^2)^N \geq 1 - 4 N \epsilon^2$ for $0<\epsilon < \frac{1}{2}$ by Bernoulli's inequality. 
Combining Equations~\ref{supp eqn: Helstrom bound} and \ref{supp eqn: trace distance derivation}, Helstrom's bound stipulates $N\geq \epsilon^{-2} \left(\frac{1}{2}-\delta\right)^2$.
\end{proof}

If we give the success probability a number, e.g. $1-\delta=\frac{2}{3}$, then the lower bound is explicitly $N\geq \epsilon^{-2}/36\approx 0.0278\epsilon^{-2}$. Within the asymptotic big-$\Omega$ notation, the lower bound is equivalent to the sample complexity of Theorem 1 in main text, $N=\Omega(\epsilon^{-2})$, thereby proving the scheme by BSs and PNRDs achieve the optimal complexity up to a constant scaling. 

For any method that solves the overlap estimation problem, the special states we defined above, $\hat{\sigma}_{\pm}^{\otimes N}$, constitute worst-case instances where the method cannot guarantee a sample complexity better than $N=O(\epsilon^{-2})$ without violating the Helstrom's bound. Therefore, the sample complexity provided by our joint overlap estimation method is optimal up to a constant scaling within the big-$O$ notation. 

Note that throughout the proof for Theorem~\ref{supp theorem: overlap estimation lower bound}, no assumptions were made on the type of measurements, which can include any arbitrary combination of multi-copy, entangling or adaptive measurements. The fact that our joint overlap estimation scheme achieves the optimal scaling by use of just linear optics and PNRDs is a further advantage. 

The proof was originally given in Ref.~\cite{Anshu-2022-distributed_quantum_inner_product} for proving the optimality of the swap test in DV systems. In CV systems, the pure states $|\phi_0\rangle, |\phi_1\rangle$ can be the squeezed vacuum state and the single-photon-subtracted squeezed vacuum state, which satisfy the orthogonality condition. 

\section{Distributed overlap estimation}
To the best of our knowledge, no rigorous sample complexity lower bound exists for distributed overlap estimation for photonic states (including CV states). We formally define the distributed overlap estimation problem, and then derive a possible method based on Monte-Carlo integration.

\begin{problem}\label{supp problem: distributed overlap estimation}
    \textup{\textbf{Distributed overlap estimation.}} Alice is given some $N$-copy state $\left(\hat{\rho}^{(A)}\right)^{\otimes N}$, and Bob is given $\left(\hat{\rho}^{(B)}\right)^{\otimes N}$. Given some error $\epsilon\in(0,\frac{1}{2})$, their goal is to estimate $\Tr\left[\hat{\rho}^{(A)} \hat{\rho}^{(B)}\right]$ to within additive error $\epsilon$ with success probability at least $2/3$ by only local quantum operations and classical communication (LOCC).  
\end{problem}

For DV states $\hat{\rho}^{(A/B)}$ acting on $d$-dimensional Hilbert space, Ref.~\cite{Anshu-2022-distributed_quantum_inner_product} proved a rigorous lower bound $N=\Omega(\max(\epsilon^{-2}, \sqrt{d} \epsilon^{-1})$, even if adaptive multi-copy measurements and arbitrary rounds of classical communication are allowed. The proof works by converting the distributed overlap estimation problem into a decisional problem, similar to what we have done in the previous section. 

The decisional problem of distributed overlap estimation probem can be described as follows: Alice has a $d$-dimensional Haar random state, $\hat{\rho}$, and Bob with equal prior probability either has the same state $\hat{\rho}$ or an independent $d$-dimensional Haar random state, $\hat{\sigma}$. Locally, Alice and Bob cannot tell the difference between the two situations. However, if Alice and Bob can solve the distributed overlap estimation problem by LOCC, then they can discriminate which situation they are in with high probability. This is because if Bob also has the same $\hat{\rho}$, then their overlap is $1$, whereas if Bob has a different $\hat{\sigma}$, the overlap scales as $1/d$. Ref.~\cite{Anshu-2022-distributed_quantum_inner_product} proves an upper bound for the success probability of this decisional problem, which lower bounds the sample complexity of the distributed overlap estimation problem. 

It is not immediately obvious how to extend the proof in Ref.~\cite{Anshu-2022-distributed_quantum_inner_product} to CV systems. However, since energy-constrained CV states are approximately DV states, we conjecture that the same curse of dimensionality equally applies to CV systems. 

In what follows, we consider a method for distributed overlap estimation on CV states in a special case when CV states are easy to learn. Our methods serve to upper bound the sample complexity for energy-constrained CV states. We leave the proof of rigorous lower bound for future work. 

As shown in Equation~\ref{supp eqn: overlap in characteristic function form}, the overlap between two CV states are given by the overlap integral of their characteristic functions. Therefore, our method works by estimating the characteristic functions of the two states, $\chi^{(A/B)}(\bm{\alpha})$, at $L$ points in phase space, before performing Monte-Carlo integration of the estimated values to estimate the overlap. 

In general, learning the characteristic function of CV states is inefficient. Ref~\cite{coroi2025exponentialadvantage} proved that in the worst-case, estimation of $\chi^{(A/B)}(\bm{\alpha})$, even at just one point in phase space and allowing multi-copy entangling measurements, requires at least an exponential number of copies of the state. Therefore by extension, any general distributed overlap estimation algorithm by individually learning the characteristic functions must also require exponential sample complexity. 

However, for a special case of self-reflective states, learning their characteristic function can be efficient. These types of states are defined as follows~\cite{Wu-2024-efficient_learning, coroi2025exponentialadvantage}: 
\begin{definition}
    An $M$-mode state $\hat{\rho}$ with characteristic function $\chi(\bm{\alpha})$ is self-reflective if there exists some symmetric unitary matrix $\bm{U}\in\mathds{C}^{M\times M}$ such that $\chi(\bm{\alpha})=\chi(\bm{U\alpha}^*)$. 
\end{definition}

The symmetric unitary matrix $\bm{U}$ defines a set of axes in phase space that passes through the origin point, against which the state is reflected. For a state with reflection symmetry, there exists some set of axes against which the reflection is the state itself. A large class of interesting CV states are self-reflective, such as coherent states, squeezed states, Fock states, Gottesman-Kitaev-Preskill (GKP) states and Schr{\"o}dinger cat states~\cite{Wu-2024-efficient_learning, coroi2025exponentialadvantage}. However, there are also simple states that do not exhibit reflection symmetry, such as displaced squeezed states whose displacement axis and squeezing axis do not align. 

For self-reflective states, Ref.~\cite{Wu-2024-efficient_learning} proved the following lemma for efficient learning:  
\begin{lemma}\label{supp lemma: self reflective states characteristic function estimation}
    \textup{\textbf{(Theorem 1 of Ref~\cite{Wu-2024-efficient_learning}.)}}    
   For a self-reflective state $\hat{\rho}$ with characteristic function $\chi(\bm{\alpha})$, given $\hat{\rho}^{\otimes N}$, $L$ points in phase space, $\{\bm{\alpha}_i\}_{i=1}^{L}$, and some $\tilde{\epsilon}, \tilde{\delta}\in(0,\frac{1}{2})$, the values of $\{\chi(\bm{\alpha}_i)\}_{i=1}^L$ can be estimated to within additive error $\tilde{\epsilon}$ with success probability at least $1-\tilde{\delta}$ for sample complexity $N=O\left(\tilde{\epsilon}^{-4} \log(L/\tilde{\delta})\right)$. 
\end{lemma}

Lemma \ref{supp lemma: self reflective states characteristic function estimation} overcomes the first obstacle of efficient characteristic function learning. In fact, it enables estimation of the characteristic function at exponential number of points in phase space using only a polynomial number of measurements. 

A second obstacle is that the functions $\chi^{(A/B)}(\bm{\alpha})$ are theoretically defined over an infinite phase space, making Monte-Carlo integration impossible. However, in physical relevant contexts, the CV states are invariably bounded by energy, which makes it possible to evaluate the overlap integral in Equation~\ref{supp eqn: overlap in characteristic function form} within a finite volume of phase space. We make the assumption that the functions $\chi^{(A/B)}(\bm{\alpha})$ are only important for $|\bm{\alpha}|^2\leq \kappa M$, where $M$ is the mode number and $\kappa$ is an energy constraint per mode, inspired by similar techniques in theoretical quantum optics literature~\cite{coroi2025exponentialadvantage, Wu-2024-efficient_learning, mele2024learningquantumstatescontinuous}. 

To be more specific, we assume that the overlap is bounded within a $2M$-dimensional hypersphere in phase space that has radius $|\bm{\alpha}|=\sqrt{\kappa M}$, such that: 
\begin{equation}\label{supp eqn: bound on overlap phase space}
    \left|
    \frac{1}{\pi^M} \int_{|\bm{\alpha}|^2 \geq \kappa M} 
    d^{2M} \bm{\alpha} \chi^{(A)}(\bm{\alpha}) \chi^{(B)}(\bm{-\alpha})
    \right|
    \leq \frac{\epsilon}{2}.
\end{equation}
If we denote the hypersphere as $\mathcal{A}$, its volume is given by the hypersphere volume formula: 
\begin{equation}
    |\mathcal{A}|=\int_{|\bm{\alpha}|^2 \leq \kappa M}  d^{2M} \bm{\alpha}
    =\frac{(\pi \kappa M)^M}{M!} 
    = O\left(\frac{(\mathrm{e}\pi\kappa)^M}{\sqrt{2\pi M}}\right), 
\end{equation}
where the asymptotic scaling is derived by using the Stirling's approximation. 

Therefore, we devise the following method to solve the distributed overlap estimation problem, if the states are \textit{a priori} known to be self-reflective: 
\begin{enumerate}
    \item Alice select $L$ points in the hypersphere $\mathcal{A}$, $\{\bm{\alpha}_i\}_{i=1}^L$, that satisfy $|\bm{\alpha}_i|^2 \leq \kappa M$, and communicates the selection to Bob.
    
    \item Alice and Bob each uses Lemma~\ref{supp lemma: self reflective states characteristic function estimation} to learn estimates to the characteristic function of their respective state, $\tilde{\chi}^{(A/B)}$, such that
    \begin{equation}\label{supp eqn: distributed chi errors}
        \Pr\left[
        \bigcup_{i=1}^L 
        \left|\tilde{\chi}^{(A)}(\bm{\alpha}_i) - \chi^{(A)}(\bm{\alpha}_i)\right| \geq \tilde{\epsilon}
        \right] \leq \tilde{\delta}
        ,\quad \textup{and}\quad
        \Pr\left[
        \bigcup_{i=1}^L 
        \left|\tilde{\chi}^{(B)}(-\bm{\alpha}_i) - \chi^{(B)}(-\bm{\alpha}_i)\right| \geq \tilde{\epsilon}
        \right] \leq \tilde{\delta}.
    \end{equation}

    \item They approximate the overlap expression in Equation~\ref{supp eqn: overlap in characteristic function form} using the unbiased estimator $\tilde{Y}$, defined as
    \begin{equation}\label{supp eqn: Monte carlo overlap estimator by chi}
        \tilde{Y}
        =
        \frac{|\mathcal{A}|}{\pi^M L}
        \sum_{i=1}^L
        \tilde{\chi}^{(A)}(\bm{\alpha}_i)
        \tilde{\chi}^{(B)}(-\bm{\alpha}_i).
    \end{equation}
\end{enumerate}

Using techniques from Ref.~\cite{Wu-2024-efficient_learning}, we can bound the estimation error by,  
\begin{equation}\label{supp eqn: derive bound on overlap error by chi monte carlo}
\begin{split}
    \left|\tilde{Y} - \Tr[\hat{\rho}^{(A)} \hat{\rho}^{(B)}]\right|
    =&
    \frac{1}{\pi^M} 
    \left|
    \int_{\mathds{C}^M} 
    d^{2M} \bm{\alpha}\, 
    \chi^{(A)}(\bm{\alpha}) \chi^{(B)}(\bm{-\alpha})
    -
    \frac{|\mathcal{A}|}{L}
    \sum_{i=1}^L
    \tilde{\chi}^{(A)}(\bm{\alpha_i})
    \tilde{\chi}^{(B)}(-\bm{\alpha_i})
    \right|
    \\
    \leq& 
    \frac{1}{\pi^M} 
    \left|
    \int_{\mathcal{A}} 
    d^{2M} \bm{\alpha}\, 
    \chi^{(A)}(\bm{\alpha}) \chi^{(B)}(\bm{-\alpha})
    -
    \frac{|\mathcal{A}|}{L}
    \sum_{i=1}^L
    \tilde{\chi}^{(A)}(\bm{\alpha_i})
    \tilde{\chi}^{(B)}(-\bm{\alpha_i})
    \right|
    +\frac{\epsilon}{2}
    \\
    \leq& 
    \frac{1}{\pi^M} 
    \left|
    \int_{\mathcal{A}} 
    d^{2M} \bm{\alpha}\, 
    \chi^{(A)}(\bm{\alpha}) \chi^{(B)}(\bm{-\alpha})
    -
    \frac{|\mathcal{A}|}{L}
    \sum_{i=1}^L
    {\chi}^{(A)}(\bm{\alpha_i})
    {\chi}^{(B)}(-\bm{\alpha_i})
    \right|
    \\
    &+ 
    \frac{|\mathcal{A}|}{L \pi^M}
    \left|
    \sum_{i=1}^L
    \left[
    \chi^{(A)}(\bm{\alpha}_i)\chi^{(B)}(-\bm{\alpha}_i)
    -
    \tilde{\chi}^{(A)}(\bm{\alpha_i})
    \tilde{\chi}^{(B)}(-\bm{\alpha_i})
    \right]
    \right|
    +\frac{\epsilon}{2}
\end{split}
\end{equation}
The first term in the last line is the estimation error in Monte-Carlo integration, given by $\frac{\sigma_L |\mathcal{A}|}{\pi^M\sqrt{L}}$, where $\sigma_L$ is the population variance estimated by the sample variance:
\begin{equation}
    (\sigma_L)^2 
    =
    \frac{1}{L-1}
    \sum_{i=1}^L
    \left[
    \chi^{(A)}(\bm{\alpha}_i)
    \chi^{(B)}(-\bm{\alpha}_i)
    - \frac{1}{L}
    \left(
    \sum_{j=1}^L
    \chi^{(A)}(\bm{\alpha}_j)
    \chi^{(B)}(-\bm{\alpha}_j)
    \right)
    \right]. 
\end{equation}
As long as this error remains finite, the Monte-Carlo integration error decreases monotonically as $1/\sqrt{L}$. However, in experiment we do not know the true function values, and therefore can only approximate the variance by the estimated function values:
\begin{equation}
    (\tilde{\sigma}_{L})^2
    =
    \frac{1}{L-1}
    \sum_{i=1}^L
    \left[
    \tilde{\chi}^{(A)}(\bm{\alpha}_i)
    \tilde{\chi}^{(B)}(-\bm{\alpha}_i)
    - \frac{1}{L}
    \left(
    \sum_{j=1}^L
    \tilde{\chi}^{(A)}(\bm{\alpha}_j)
    \tilde{\chi}^{(B)}(-\bm{\alpha}_j)
    \right)
    \right]. 
\end{equation}

The third term in the last line of Equation~\ref{supp eqn: derive bound on overlap error by chi monte carlo} comes from the energy constraint assumption in Equation~\ref{supp eqn: bound on overlap phase space}. 

Finally, to bound the second term in the last line of Equation~\ref{supp eqn: derive bound on overlap error by chi monte carlo}, we note that the magnitude of the characteristic function itself is bounded by the H{\"o}lder inequality for Schatten norms:
\begin{equation}
    |\chi(\bm{\alpha})|
    = 
    \Tr\left[\hat{\rho} \hat{D}(\bm{\alpha})\right]
    \leq ||\hat{D}(\bm{\alpha})||_{\infty} ||\hat{\rho}||_1 
    = 1,
\end{equation}
The operator (or spectral) norm is the largest singular value, which for a unitary operator is $||\hat{D}(\bm{\alpha})||_{\infty}=1$, and the trace norm of a density matrix is $||\hat{\rho}||_1=1$. 
This allows us to write, 
\begin{equation}
\begin{split}
    &\left| 
        \sum_{i=1}^L \left[
            \chi^{(A)}(\bm{\alpha}_i) \chi^{(B)}(-\bm{\alpha}_i) - 
            \tilde{\chi}^{(A)}(\bm{\alpha}_i) \tilde{\chi}^{(B)}(-\bm{\alpha}_i) 
        \right]
    \right|
    \\
    &\leq 
    \left|
        \sum_{i=1}^L 
        \chi^{(A)}(\bm{\alpha}_i) 
        \left[\chi^{(B)}(-\bm{\alpha}_i) - 
            \tilde{\chi}^{(B)}(-\bm{\alpha}_i) 
        \right]
    \right|
    +
    \left|
        \sum_{i=1}^L 
        \left[
            \chi^{(A)}(\bm{\alpha}_i) - 
            \tilde{\chi}^{(A)}(\bm{\alpha}_i)
        \right]
        \tilde{\chi}^{(B)}(-\bm{\alpha}_i) 
    \right|
    \\
    &\leq 
    \left|
        \sum_{i=1}^L
        \left[\chi^{(B)}(-\bm{\alpha}_i) - 
            \tilde{\chi}^{(B)}(-\bm{\alpha}_i) 
        \right]
    \right|
    \max\left(
    \left|
    \chi^{(A)}(\bm{\alpha}_i) \right|
    \right)
    + 
    \left|
        \sum_{i=1}^L 
        \left[
            \chi^{(A)}(\bm{\alpha}_i) - 
            \tilde{\chi}^{(A)}(\bm{\alpha}_i)
        \right]
    \right|
    \max\left(\left|
    \tilde{\chi}^{(B)}(-\bm{\alpha}_i) 
    \right|\right)\\
    &\leq 
    2 L \tilde{\epsilon}
\end{split}
\end{equation}
This inequality holds with probability at least $1-2\tilde{\delta}$ by applying the union bound on Equation~\ref{supp eqn: distributed chi errors}. 

Combining all three terms in Equation~\ref{supp eqn: derive bound on overlap error by chi monte carlo}, and if we require the overlap estimator to be within $\epsilon$ additive error with probability at least $1-\delta$, we have: 
\begin{equation}
    \epsilon\approx 
    \frac{\tilde{\sigma}_L |\mathcal{A}|}{\sqrt{L}\pi^M}
    + \frac{2|\mathcal{A}|\tilde{\epsilon}}{\pi^M}
    +\frac{\epsilon}{2}, 
    \quad
    \delta = 2\tilde{\delta},
\end{equation}
where $\approx$ is used because $\tilde{\sigma}_K$ is used to estimate the Monte Carlo estimation error. 
The distributed overlap estimation problem is solved if we choose: 
\begin{equation}
    \tilde{\epsilon}
    = \frac{\pi^M}{8|\mathcal{A}|}\epsilon
    , \quad 
    L = 
    \left(
    \frac{4 \tilde{\sigma}_L |\mathcal{A}|}{\pi^M \epsilon}
    \right)^2
    , \quad 
    \tilde{\delta} = \frac{1}{2}\delta.
\end{equation}
If we substitute this into the sample complexity of Lemma~\ref{supp lemma: self reflective states characteristic function estimation}, we get sample complexity,
\begin{equation}
\begin{split}
    N 
    =
    O\left(\tilde{\epsilon}^{-4} \log(L/\tilde{\delta})\right)
    &=
    O\left(
    \epsilon^{-4}
    \left(\frac{|\mathcal{A}|}{\pi^M}\right)^{4}
    \log 
    \left(\left(
    \frac{ |\mathcal{A}|}{\pi^M}
    \right)^2
    \tilde{\sigma}_L^2
    \epsilon^{-2}
    \delta^{-1}\right)
    \right)
    \\
    &=
    O\left(
    \epsilon^{-4}
    \left(\frac{(\mathrm{e}\kappa)^M}{\sqrt{M}}\right)^{4}
    \log 
    \left(
    \left(\frac{(\mathrm{e}\kappa)^M}{\sqrt{M}}\right)^2
    \tilde{\sigma}_L^2
    \epsilon^{-2}
    \delta^{-1}\right)
    \right)
\end{split}
\end{equation}
where in the last line we substituted the hypersphere formula for $|\mathcal{A}|$. 

Formally, we have the following result:
\begin{theorem}
    For distributed overlap estimation on self-reflective states, $\hat{\rho}^{(A)}$ and $\hat{\rho}^{(B)}$, given $N$ copies of each state respectively and $\epsilon, \delta\in(0, \frac{1}{2})$, if the states satisfy 
    $$
        \left|
        \frac{1}{\pi^M} \int_{|\bm{\alpha}|^2 \geq \kappa M} 
        d^{2M} \bm{\alpha} \chi^{(A)}(\bm{\alpha}) \chi^{(B)}(\bm{-\alpha})
        \right|
        \leq \frac{\epsilon}{2}
    $$
    for some energy constraint $\kappa>0$, then Alice and Bob can compute the following estimator by LOCC, 
    $$
        \tilde{O}
        =
        \frac{(\kappa M)^M}{M!} L^{-1}
        \sum_{i=1}^L
        \tilde{\chi}^{(A)}(\bm{\alpha}_i)
        \tilde{\chi}^{(B)}(-\bm{\alpha}_i),
    $$
    which satisfies 
    $
    \Pr \left[
    \left| \tilde{O} - \Tr[\hat{\rho}^{(A)} \hat{\rho}^{(B)}] \right|
    \gtrsim \epsilon
    \right] \leq \delta$,
    where $\sim$ comes from the approximation of Monte-Carlo integration error, if 
    $$
        N=O\left(
    \epsilon^{-4}
    \left(\frac{(\mathrm{e}\kappa)^M}{\sqrt{M}}\right)^{4}
    \log 
    \left(
    \left(\frac{(\mathrm{e}\kappa)^M}{\sqrt{M}}\right)^2
    \tilde{\sigma}_L^2
    \epsilon^{-2}
    \delta^{-1}\right)
    \right).
    $$
\end{theorem}

Therefore, if $\kappa \geq \mathrm{e}^{-1}$, this is still an exponential sample complexity, despite the efficient protocol for learning $\chi^{(A/B)}(\bm{\alpha})$ themselves. 

\section{Photon source}

Our experimental apparatus consists of three major components: a spontaneous parametric downconversion (SPDC) photon source, a chip-based photonic circuit and eight superconducting nanowire single-photon detectors (SNSPDs) for photon detection. In this section we give an overview of the photon source, before focusing in the next section on the details of the photonic chip. 

Our SPDC source consists of a periodically-poled potassium titanyl phosphate (ppKTP) waveguide pumped by Spectra-Physics Mai Tai, a femtosecond, wideband, mode-locked Ti:Sapphire Laser. The pump pulses have repetition rate of $80$~MHz. The spectrum of the pump pulses are filtered to $775\pm 0.9$~nm full-width half-maximum using a pair of 3~nm bandpass filters (Thorlabs FBH780-3). The type-II SPDC process on the ppKTP waveguide (AdvR Inc.) generates a pair of collinear, orthogonally polarised signal and idler modes at $1550$~nm. After the pump field is removed by a dichroic mirror followed by optical spectral filters, the signal and idler modes are separated by a polarising beamsplitter (PBS) and coupled into polarisation-maintaining (PM) fibres, which are connected to the photonic chip. The output from the photonic chip are sent by single-mode fibres into $8$ SNSPDs for photon detection. 

The signal and idler modes form a two-mode squeezed vacuum (TMSV) state, which can be described by 
\begin{equation}\label{supp eqn: tmsv}
    |\textup{TMSV}\rangle = \sqrt{1-|\lambda|^2} \sum_{n=0}^{\infty} \lambda^{n} |n\rangle_s |n\rangle_i,
\end{equation}
where $\lambda$ is a parameter that depends on the pump field and the ppKTP crystal, and subscripts $s, i$ are used to denote the signal and idler modes. 

In our experiment, we keep the pump field low at $65\mu$W. The total coincidence rate between the signal and idler modes, sent through the photonic chip and detected on SNSPDs, is measured to be $5500$ coincidences per second. To measure the higher photon number events, we utilise a Hanbury Brown and Twiss style interferometer. We block either signal or idler mode, and split the other mode on a balanced beamsplitter before detection on SNSPDs. The coincidence rate measures the probability of having $|n\rangle_s|n\rangle_i$ components with $n\geq 2$ in Equation~\ref{supp eqn: tmsv}. At the current pump intensity, we measure this coincidence rate to be lower than $50$ per second. 

Neglecting the low-rate high-photon-number events, our SPDC source can approximate a probabilistic two-photon source. During the experiment, we only post-select outcomes where two photons are measured on the detectors. By this post-selection, we can also remove the effects of photon loss in the experiment. 

The indistinguishability of our two-photon source is measured by a Hong-Ou-Mandel (HOM) sweep. The fibre coupler for the signal mode is placed on top of a piezo-electric translation stage. The signal and idler modes are interfered on a balanced beamsplitter implemented on the photonic chip. The output coincidences are measured against translation stage position. The visibility, $V=(P_{\textup{max}}-P_{\textup{min}})/{P_{\textup{max}}}$, of the HOM dip was measured to be $95.40\%$, as shown in Fig.~\ref{supp fig: HOM}. 

\begin{figure*}[h]
    \includegraphics[width=0.9\textwidth]{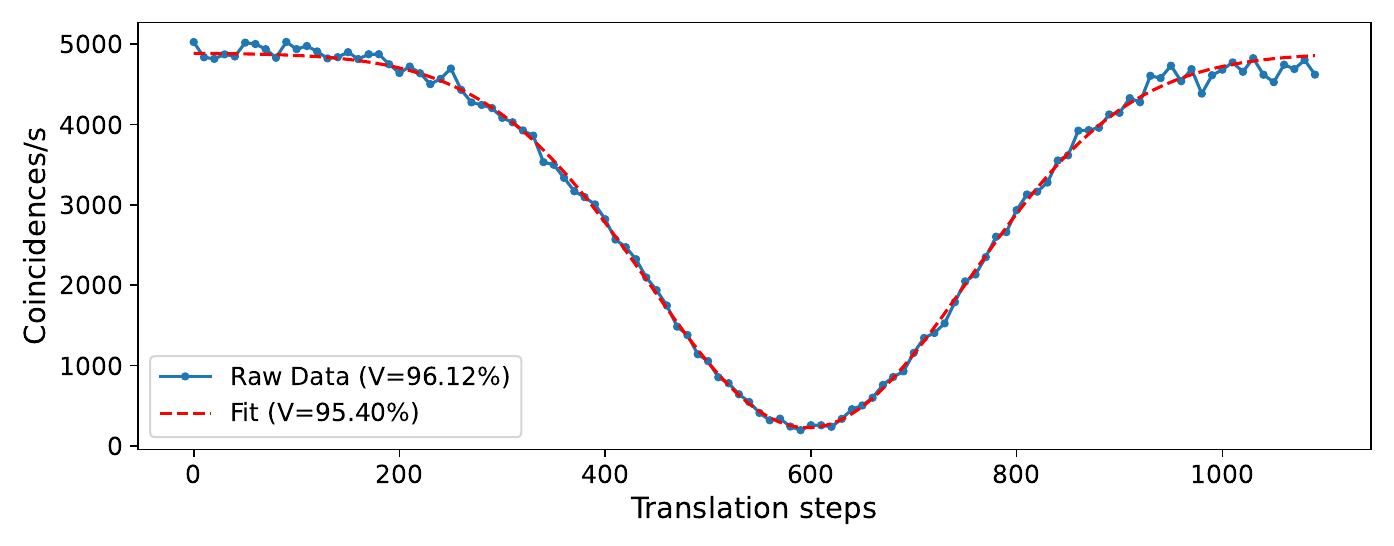}
    \caption{\label{supp fig: HOM} HOM sweep result. The visibility directly calculated from the raw coincidence data differ slightly from the fitting result to a Gaussian model. }
\end{figure*}

\section{Photonic integrated circuit}

\subsection{Setup}
The photonic integrated circuit (or chip) used in this work was fabricated by a multi-project wafer run at LioniX International BV. The TriPleX waveguide technology was used to fabricate waveguides based on alternating layers of silicon nitride and silicon dioxide~\cite{roeloffzenLowLossSi3N4TriPleX2018}. A reconfigurable linear optical circuit is implemented on the chip by use of thermo-optically tunable phaseshifters. These are sections of the waveguides with refractive index that can be thermally tuned to change the relative phase that photons pick up when passing through. In our chip, a $1$~mm long platinum resistor is placed on top of each phaseshifter and thermal tuning is achieved by passing current through the resistor. The power dissipation per phaseshifter is around $700$~mW, or around $12$~V and $60$~mA, for a $2\pi$ phase shift. This adds up to $70$~W maximum power consumption for a total of $100$ phaseshifters on the chip, though typically in practice we operate the chip at much lower power. 

The chip is optically packaged with polarisation-maintaining (PM) fibre arrays for both input and output coupling, provided by PHIX Photonics Assembly BV. Light is coupled from these fibre arrays into the chip by edge coupling. Coupling efficiency is measured to be as high as $88$\% per facet. The end-to-end transmission varies between different input and output modes and is typically around $50$\%, which includes effects of propagation loss and coupling loss. 

The remaining electrical and electronic packaging was done in-house. We use a $120$-channel source measurement system (XPOW-120AX-CCvCV-U by Nicslab Ops, Inc.) to supply independent voltage/current sourcing and measurement to the thermo-optical phaseshifters on the chip. We operate it in constant current mode, which has a $16$-bit resolution for each channel, and set the output range per channel to $0-100$~mA and $0-15$~V. The system is connected to a computer via USB and is controlled by a self-written Python package for calibration and control. 

The chip itself is mounted on a gold-plated copper base for thermal dissipation. Also mounted on the base is a printed circuit board for electrical connections. Inside the copper base is a thermistor for temperature sensing, which feeds back to a thermoelectric temperature controller (TEC) (Model 350B by Newport Corporation). Beneath the chip mount is a Peltier cooler and a copper heat sink. The Peltier cooler works by maintaining a constant temperature difference between its upper and lower surface. A fan is used to blow air through the copper heat sink for better heat dissipation. The TEC has maximum output power of $55$~W. While slightly short of the $70$~W maximum power dissipation power on chip, it suffices for the operations in this work. When in operation, we set the TEC to maintain the thermistor inside the copper base at $28.0\pm 0.1^{\circ}$C temperature. 

The chip implements a $10$-mode universal linear-optical interferometer that follows the Bell decomposition scheme~\cite{bell-2021-further_compact}. The basic unit for control in the Bell scheme is a symmetric Mach-Zehnder Interferometer (MZI) as shown in Fig.~\ref{supp fig: chip scheme}(c). Each MZI consist of two directional couplers (DCs) acting as balanced beamsplitters and two internal phaseshifters in between, one on each arm. 

A labelled schematic of the chip is shown in Fig.~\ref{supp fig: chip scheme}(a). The optical modes are labelled from $0$ to $9$. Each phaseshifter is labelled with a unique tuple $(i,j)$, where $i,j$ run from $0$ to $9$, $i$ being the row index (which is also mode index) and $j$ being the column index. The MZIs, for convenience, share the label of their upper phaseshifter. As a result, MZI labels $(i,j)$ must satisfy $i+j$ being an even number. In the Bell-scheme, there are standalone `external' phaseshifters on alternating columns on mode $0$ and $9$ for full phase programmability. These are labelled as $(0,j)$ and $(9,j)$ for odd $j$. 

\begin{figure*}[t]
    \includegraphics[width=0.9\textwidth]{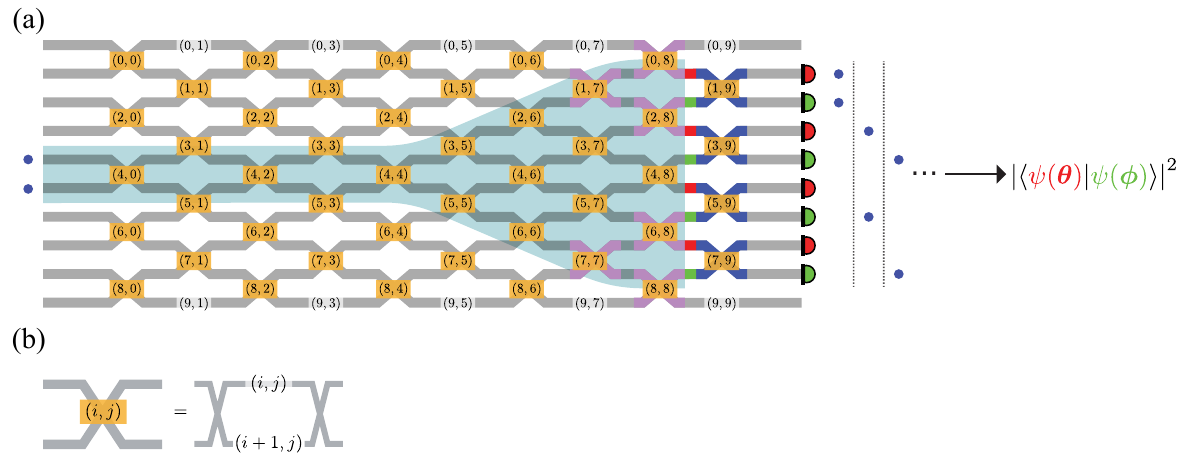}
    \caption{\label{supp fig: chip scheme}(a) Labelled scheme. MZIs and phaseshifters are labelled by orange and white background boxes respectively. Grey lines are waveguides, with DCs omitted and MZIs simplified as waveguide crossings. The six MZIs for active phase encodings are coloured by purple. After column $8$, one qudit state $|\psi(\bm{\theta})\rangle$ is encoded in modes $1,3,5,7$, coloured red; another qudit state $|\psi(\bm{\phi})\rangle$ is encoded in modes $2,4,6,8$, coloured green. The two qudit states are interfered mode-wise on four MZIs that implement balanced beamsplitters on column $9$, coloured blue. The output is measured by eight SNSPDs on modes $1$ to $8$. The output coincidence distribution directly estimates the state overlap. (b) Layout of a symmetric MZI labelled $(i,j)$, with internal phaseshifters labelled $(i,j)$ and $(i+1, j)$. }
\end{figure*}

\subsection{Control}
The action of an ideal phaseshifter applies some phase $\theta$ to a single optical mode, which can be represented by the transfer matrix
\begin{equation}
    \begin{pmatrix}
        1 & 0 & \cdots & \\ 
        0 & \ddots & &  \\ 
        \vdots &  & e^{i\phi} & \\
        & & & \ddots & \\ 
        & & & & 1 
    \end{pmatrix}. 
\end{equation}

In a thermo-optic phaseshifter, a change in temperature causes a proportionate change in refractive index depending on the thermo-optic coefficient of the material. This change is typically linear with negligible second order term~\cite{hryciwThermoopticTuningErbiumdoped2011, arbabiMeasurementsRefractiveIndices2013}. Therefore, we can assume a linear relationship between the dissipated electric power $P_{\text{el}}=VI$ in the resistor and the induced phase change $\Delta \theta$, i.e.
\begin{equation}
    \Delta \theta \propto \Delta n_{\text{eff}} \propto \Delta T \propto P_{\text{el}}, 
\end{equation}
where $\Delta n_{\text{eff}}$ is the change in effective refractive index of the waveguide segment, and $\Delta T$ is the change in temperature. The proportionality constant depends on the geometry of the waveguide and the material properties.

The XPOW source measurement system can measure both voltage and current. However, the voltage measurement does not accurately reflect the voltage drop across the phaseshifter due to non-zero resistance on the shared ground. Therefore, we operate the XPOW in constant current mode and only use the current measurement. While theoretically the disspated power should be $P_{\text{el}}=I^2 R$, the resistance of the phaseshifter also changes as it heats up and becomes a function of the current. This effect is non-negligible especially since the $2\pi$ power is as high as $700$~mW. Therefore, we write out the phase change as a polynomial function of the current: 
\begin{equation}
    \theta = b + k^{(2)} I^2 + k^{(3)} I^3 + \cdots,
\end{equation}
where $b$ is the residual phase when no current is applied. Given a set of calibration data of applied current versus phase change, a polynomial fit can be used to determine the coefficients up to a cut-off. In practice, we store the raw calibration data and use numerical interpolation to determine the required current for a target phase.

The transfer matrix of a directional coupler (DC) can be written in the form of~\cite{bandyopadhyayHardwareErrorCorrection2021}
\begin{equation}
    \begin{pmatrix}
        \cos(\frac{\pi}{4} + \alpha) & i \sin(\frac{\pi}{4} + \alpha) \\
        i \sin(\frac{\pi}{4} + \alpha) & \cos(\frac{\pi}{4} + \alpha)
    \end{pmatrix}, 
\end{equation}
where $\alpha$ is a small error term accounting for fabrication imperfections. For an ideal balanced DC, $\alpha=0$. The unitary transfer matrix of a symmetric MZI (Fig.~\ref{supp fig: chip scheme}(b)) with internal phaseshifters $\theta_1,\theta_2$ and DC errors $\alpha, \beta$ can then be derived as
\begin{equation}\label{supp eqn: MZI transfer matrix}
    \tilde{T}(\theta_1, \theta_2, \alpha, \beta)
    =
    \begin{pmatrix}
        \cos\beta & i\sin\beta \\
        i\sin\beta & \cos\beta
    \end{pmatrix}
    T(\theta_1, \theta_2)
    \begin{pmatrix}
        \cos\alpha & i\sin\alpha \\
        i\sin\alpha & \cos\alpha
    \end{pmatrix},
\end{equation}
where $T(\theta_1, \theta_2)$ is the ideal MZI transfer matrix (when $\alpha=\beta=0$), given by
\begin{equation}
    T(\theta_1, \theta_2)
    =
    i e^{i\Sigma}
    \begin{pmatrix}
        \sin\delta & \cos\delta \\
        \cos\delta & -\sin\delta
    \end{pmatrix},
\end{equation}
for MZI phase $\Sigma = (\theta_1 + \theta_2)/2$ and internal phase $\delta = (\theta_1 - \theta_2)/2$. If $\theta_{1,2}$ are fully tunable in range $[0,2\pi]$, then $\Sigma$ can cover $[0,2\pi]$ and $\delta$ can cover $[-\pi, \pi]$. The $\Sigma$ and $\delta$ phase are the two degrees of freedom for reconfiguring each MZI. The rest of the text sometimes refer to them as the `MZI phase' and `internal phase' respectively, or simply as `$\Sigma/\delta$ phase'. 

There are three special states of the MZI, namely the `cross' state, the `bar' state, and the `balanced' state, which correspond to setting the $\delta$ phase to $0$, $\frac{\pi}{2}$ or $\frac{\pi}{4}$ respectively. 

Section~\ref{supp section: calibration} will describe the calibration scheme for the chip. The goal of the calibration is to find a mapping between the current values supplied via the XPOW unit and the phases on each phaseshifter. If all phases are fully tunable between $[0, 2\pi]$, then the rectangular mesh of phaseshifters and MZIs shown in Fig.~\ref{supp fig: chip scheme}(a) is sufficient to implement any arbitrary $10\times 10$ unitary matrix up to a global phase on each mode. By adding an additional layer of stand-alone phaseshifters on the input and output modes respectively, our chip layout would recover the universal linear optical scheme proposed in Ref.~\cite{bell-2021-further_compact}. 

However, for linear-optical experiments where the input light is a multi-mode Fock state and the output of the circuit is directly measured by phase-insensitive photon detectors, the input and output phaseshifters do not change experimental outcome and can be omitted. In the next section, we will describe how an overlap estimation circuit is implemented on chip for our experiment. 

\subsection{Overlap estimation circuit}
The overlap estimation circuit is shown in Fig.~\ref{supp fig: chip scheme}(a). The input are two single photons sent into modes $4$ and $5$ respectively (modes labelled from $0$ to $9$). Hence, the input state can be described by $\hat{a}_4^\dagger \hat{a}_5^\dagger |0\rangle$, where $\hat{a}_{4,5}^\dagger$ are the creation operators for the two modes. 

We wish to encode two four-mode qudit states, which occupy modes $1,3,5,7$ and $2,4,6,8$ respectively, that are described by:  
\begin{equation}\label{supp eqn: encoded qudits}
\begin{split}
    |\bm{\theta}\rangle = |\psi(\bm{\theta})\rangle 
    &=
    \left(
    A_0 \hat{a}_{1}^\dagger + 
    A_1 e^{i\theta_1} \hat{a}_3^\dagger +
    A_2 e^{i(\theta_1+\theta_2)} \hat{a}_5^\dagger + 
    A_3 e^{i(\theta_1+\theta_2+\theta_3)} \hat{a}_7^\dagger
    \right) |0\rangle, \\
    &= 
    e^{i(\theta_1+\theta_2-\frac{\pi}{2})}
    \left(
    A_0 e^{-i(\theta_1+\theta_2-\frac{\pi}{2})} \hat{a}_{1}^\dagger + 
    A_1 e^{-i(\theta_2-\frac{\pi}{2})} \hat{a}_3^\dagger +
    A_2 e^{i\frac{\pi}{2}} \hat{a}_5^\dagger + 
    A_3 e^{i(\theta_3+\frac{\pi}{2})} \hat{a}_7^\dagger
    \right) |0\rangle
    \\
    |\bm{\phi}\rangle = |\psi(\bm{\phi})\rangle 
    &=
    \left(
    A_0 \hat{a}_{2}^\dagger + 
    A_1 e^{i\phi_1} \hat{a}_4^\dagger +
    A_2 e^{i(\phi_1+\phi_2)} \hat{a}_6^\dagger + 
    A_3 e^{i(\phi_1+\phi_2+\phi_3)} \hat{a}_8^\dagger
    \right) |0\rangle, \\
    &=
    e^{i(\phi_1+\frac{\pi}{2})}
    \left(
    A_0 e^{-i(\phi_1+\frac{\pi}{2})} \hat{a}_{2}^\dagger + 
    A_1 e^{-i\frac{\pi}{2}} \hat{a}_4^\dagger +
    A_2 e^{i(\phi_2-\frac{\pi}{2})} \hat{a}_6^\dagger + 
    A_3 e^{i(\phi_2+\phi_3-\frac{\pi}{2})} \hat{a}_8^\dagger
    \right) |0\rangle,
\end{split}
\end{equation}
where we denote the qudit states as $|\bm{\theta}\rangle, |\bm{\phi}\rangle$ as shorthand. In the second line of each qudit expression, we extracted a global phase that can be ignored in the encoding. The result is that the phases are set relative to output modes $4$ and $5$, which sit in the middle of the chip. This way, we can keep MZI $(4,8)$ at a low power setting. This separates the actively tuned phaseshifters in column $8$ and reduces the crosstalk effects. 

The amplitudes $\{A_i\}_{i=0}^3$ are fixed and chosen to be 
\begin{equation}\label{supp eqn: qudit amplitude values}
    \begin{cases}
        A_0 &= \sin(\varphi_1/2) \\
        A_1 &= \cos(\varphi_1/2) \sin(\varphi_2/2) \\
        A_2 &= \cos(\varphi_1/2) \cos(\varphi_2/2) \sin(\varphi_3/2) \\
        A_3 &= \cos(\varphi_1/2) \cos(\varphi_2/2) \cos(\varphi_3/2) 
    \end{cases}
    \quad 
    \textup{for} 
    \quad 
    \begin{cases}
        \varphi_1 &= 0.86231713 \\
        \varphi_2 &= 1.34230503 \\
        \varphi_3 &= 1.66199945
    \end{cases}.
\end{equation}
The values are chosen by optimising the Support Vector Machine's (SVM) classification accuracies for the three datasets in Section~\ref{supp section: classification task}. 

In columns $0$ to $8$ of the chip, we implement a unitary $\bm{U}$, such that 
\begin{equation}\label{supp eqn: circuit unitary expression}
\begin{split}
    \bm{U}(\bm{\theta}, \bm{\phi})\left(\hat{a}_4^\dagger \hat{a}_5^\dagger |0\rangle\right)
    =
    &\left(
    A_0 e^{i(\frac{\pi}{2}+\phi_{(0,8)} + \phi_{(1,7)})} \hat{a}_{1}^\dagger + 
    A_1 e^{i(\frac{\pi}{2}+\phi_{(2,8)} + \phi_{(1,7)})}  \hat{a}_3^\dagger +
    A_2 e^{i\frac{\pi}{2}} \hat{a}_5^\dagger + 
    A_3 e^{i(\frac{\pi}{2} + \phi_{(6,8)})} \hat{a}_7^\dagger
    \right)\\
    &\left(
    A_0 e^{i(-\frac{\pi}{2}+\phi_{(2,8)})} \hat{a}_{2}^\dagger + 
    A_1 e^{-i\frac{\pi}{2}} \hat{a}_4^\dagger +
    A_2 e^{i(-\frac{\pi}{2} + \phi_{(6,8)} + \phi_{(7,7)})} \hat{a}_6^\dagger + 
    A_3 e^{i(-\frac{\pi}{2}-\phi_{(8,8)}-\phi_{(7,7)})} \hat{a}_8^\dagger
    \right)
    |0\rangle.
\end{split}
\end{equation}
Implementation of this unitary is not unique. We adopt the circuit shown in Figure~\ref{supp fig: chip scheme}. In columns $0$ to $4$, the internal phases of MZIs on the route of the light are set to $\delta=\frac{\pi}{2}$, such that the two photons propagate almost uninterfered with each other. The internal phases of MZIs on columns $5$ to $8$ are set to implement the amplitude values $\{A_i\}_{i=0}^3$ in Equation~\ref{supp eqn: qudit amplitude values}. The MZI phases in columns $5$ and $6$ are set to equal on each active MZI. 

The relative phases $\phi_{(i,j)}$ in Equation~\ref{supp eqn: circuit unitary expression} are MZI phases relative to the column reference phase. To be more specific, the relative phase in column $8$ is the $\Sigma_{(4,8)}$ phase encoded on MZI $(4,8)$. In Column $7$, MZIs $(3,7)$ and $(5,7)$ are set to have the same $\Sigma$ phase and serve as the relative phase. Hence, by tuning the MZI phases (i.e. the $\Sigma$ phases) on the following six MZIs, we can encode arbitrary qudits in the form of Equation~\ref{supp eqn: encoded qudits}: MZIs $(0,8), (2,8), (1,7), (6,8), (8,8), (7,7)$. 


Finally, the four MZIs $(1,9), (3,9), (5,9), (7,9)$ in column 9 are set to have internal phase $\delta=\frac{\pi}{4}$ to implement balanced beamsplitters:
\begin{equation}
    T_{\textup{balanced}} \propto 
    \frac{1}{\sqrt{2}}
    \begin{pmatrix}
        1 & 1 \\ 1 & -1 
    \end{pmatrix}.
\end{equation}
The relative $\pi$ phase difference between neighbouring modes in Equation~\ref{supp eqn: circuit unitary expression} evens out the input phase difference in $T_{\textup{balanced}}$. 

The output modes from $1$ to $8$ are connected to eight SNSPDs. We measure and post-select two-photon outputs to trace out the effects of loss. Splitting the output modes into two registers, we showed in the main text that coincidence probability across the two registers allows us to estimate the overlap between the qudit states. In our scheme, the two registers happen to be odd-number-labelled modes $1,3,5,7$ and even-number-labelled modes $2,4,6,8$. Therefore, we divide the two-photon outputs into even-parity events, i.e. coincidence events between modes $k$ and $l$ where $k+l$ is an even number; and odd-parity events where $k+l$ is an odd number. If we collect $N$ total samples that measure two photons, which include $N_\textup{odd}$ odd-parity events, the overlap is estimated by 
\begin{equation}
    \left|
    \langle \psi(\bm{\theta})| \psi(\bm{\phi})\rangle \right|^2
    \approx 
    1 - 2(1-R) \frac{N_\textup{odd}}{N},
\end{equation}
where $R=\sum_{i=0}^3 A_i^4 < 1$. This equals the bunching probability, or probability of having two photons in the same output mode that cannot be distinguished by click detectors such as SNSPDs. If photon-number resolving detectors (PNRDs) were used, then we do not need this rescaling factor.

\subsection{Calibration}\label{supp section: calibration}

\begin{figure*}[t]
    \includegraphics[width=0.9\textwidth]{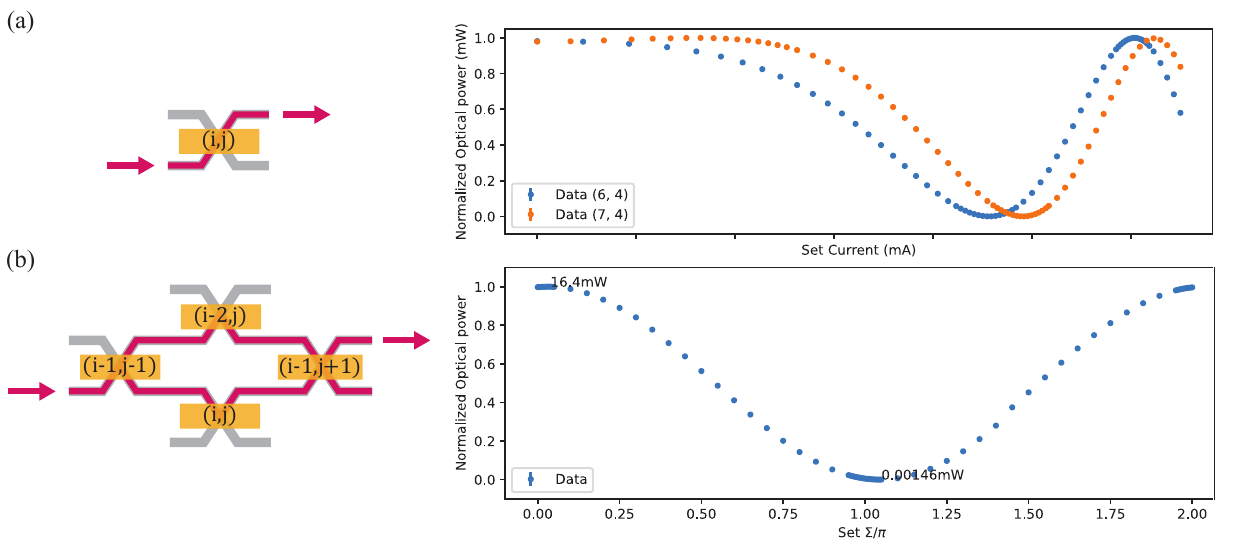}
    \caption{\label{supp fig: chip calibration} (a) Calibration of internal phaseshifters and the $\delta$ phase of MZI $(i,j)$. (b) Calibration of the $\Sigma$ phase of MZI $(i,j)$. }
\end{figure*}

The goal of calibration is to find a mapping between the current values applied to the electric channels and the phase values of each MZI and phaseshifter. 

For an MZI described by Equation~\ref{supp eqn: MZI transfer matrix}, if we send a coherent state with unit intensity into the lower input mode, then power in the upper output mode is: 
\begin{equation}
    P = \sin^2(\beta-\alpha) + 
    \frac{\cos^2(\beta+\alpha) - \sin^2(\beta-\alpha)}{2} 
    \left(\cos(\theta_1-\theta_2)+1\right). 
\end{equation}
We label the two internal phaseshifters as $(i,j)$ and $(i+1,j)$, and express the internal phases as 
\begin{equation}
    \theta_{i,j} = \Delta_{i,j}(I_{i,j})+b_{i,j}, 
    \quad 
    \theta_{i+1,j} = \Delta_{i+1,j}(I_{i+1,j})+b_{i+1,j}, 
\end{equation}
where $b$ is the residual phase when no current is applied, and $\Delta(I)$ is the phase change that is a function of the applied current $I$. 
If we sweep the current on each internal phaseshifter, the output power follows a sinusoidal curve as shown in Fig.~\ref{supp fig: chip calibration}(a). From this, we reconstruct the relationship between $\Delta(I)$ and $I$ for each internal phaseshifter, as well as the residual relative phase, $b_{i,j}-b_{i+1,j}$. This enables us to set the internal $\delta$ phase of the MZI. 

However, to set the $\Sigma$ phase of the MZI, we also need to know $b_{i,j}+b_{i+1,j}$ relative to the next MZI in the column. To do so, we construct a `meta-MZI' as shown in Fig.~\ref{supp fig: chip calibration}(b) using four MZIs, which are labelled $(i-2,j), (i,j), (i-1,j-1), (i-1,j+1)$. The target is MZI $(i,j)$, which is set to bar state that fixes $\delta_{i,j}=\frac{\pi}{2}$ and leaves $\Sigma_{i,j}$ as the free parameter to be calibrated. The MZI $(i-2, j)$ is also set to cross state, with $\delta_{i-2,j}=0$ and $\Sigma_{i-2,j}$ fixed. The MZIs $(i-1,j-1)$ and $(i-1,j+1)$ are set to balanced states, with $\delta_{i-1,j-1}=\delta_{i-1,j+1}=\frac{\pi}{4}$ and $\Sigma_{i-1,j-1}, \Sigma_{i+1,j-1}$ arbitrary as they do not affect the output power.

If light is sent into the lower input port of MZI $(i-1, j-1)$, then the upper output power of MZI $(i-1, j+1)$ can be modelled as 
\begin{equation}
    P = a + c \cos\left(\Sigma_{i,j} - \Sigma_{i-2,j} + \xi\right)
\end{equation}
where $a, c$ are constants that depend on the errors of the effective DCs as well as input power, channel efficiency and detector efficiency. The phase $\xi$ is an error term that's due to the DC errors inside MZI $(i-1,j-1)$ and MZI $(i-1,j+1)$. Imperfect DCs inside an MZI results in a small relative phase between its two outputs, which become a phase error in the meta-MZI. By sweeping the relative MZI phase, $\Sigma_{(i,j)}-\Sigma_{(i-2,j)}$, we get a sinusoidal output power curve as shown in Fig.~\ref{supp fig: chip calibration}(b), whose offset phase estimates the relative residual phase $b_{i,j}+b_{i+1,j} - b_{i-2,j} - b_{i-1,j}$. This enables $\Sigma$ phase tuning on MZI $(i,j)$. 

Additionally, we can also calibrate the directional coupler errors of the chip from the visibilities of the sinusoidal sweeps. Among the active MZIs used in this experiment, the median splitting ratio of directional couplers is $52.8:47.2$. 

Relative channel efficiencies can be estimated by directing light into different output modes by setting the appropriate MZIs and comparing the output power, accounting for the MZI encoding errors. 

Finally, crosstalk effects between MZIs were calibrated by repeating calibrations at various phase settings of neighbouring MZIs. The calibration results can be interpolated to reconstruct a look-up table that describes phase changes due to crosstalk for each location and and source of crosstalk. 

After calibration, the resulting circuit fidelity is shown in Fig.~\ref{supp fig: chip output}. Fig.~\ref{supp fig: chip output}(a) shows the fidelity of the output coincidence distribution: $F=\sum_{(i,j)} \sqrt{P^{\textup{(ideal)}}_{(i,j)} P^{\textup{(exp)}}_{(i,j)}}$, where $P_{(i,j)}$ is the coincidence probability between modes $i$ and $j$, and superscripts $\textup{(ideal)}$ and $\textup{(exp)}$ indicate ideal calculations and experimental estimates, respectively. The average fidelity is $0.980\pm0.006$ across the experiments conducted in this work. 

The resulting overlap estimation accuracy is visualised in Fig.~\ref{supp fig: chip output}(b). As is clear from the fidelity values as well as the overlap values, there is still room for improvement in the calibration scheme, especially for better crosstalk mitigation.

\begin{figure*}[t]
    \includegraphics[width=0.9\textwidth]{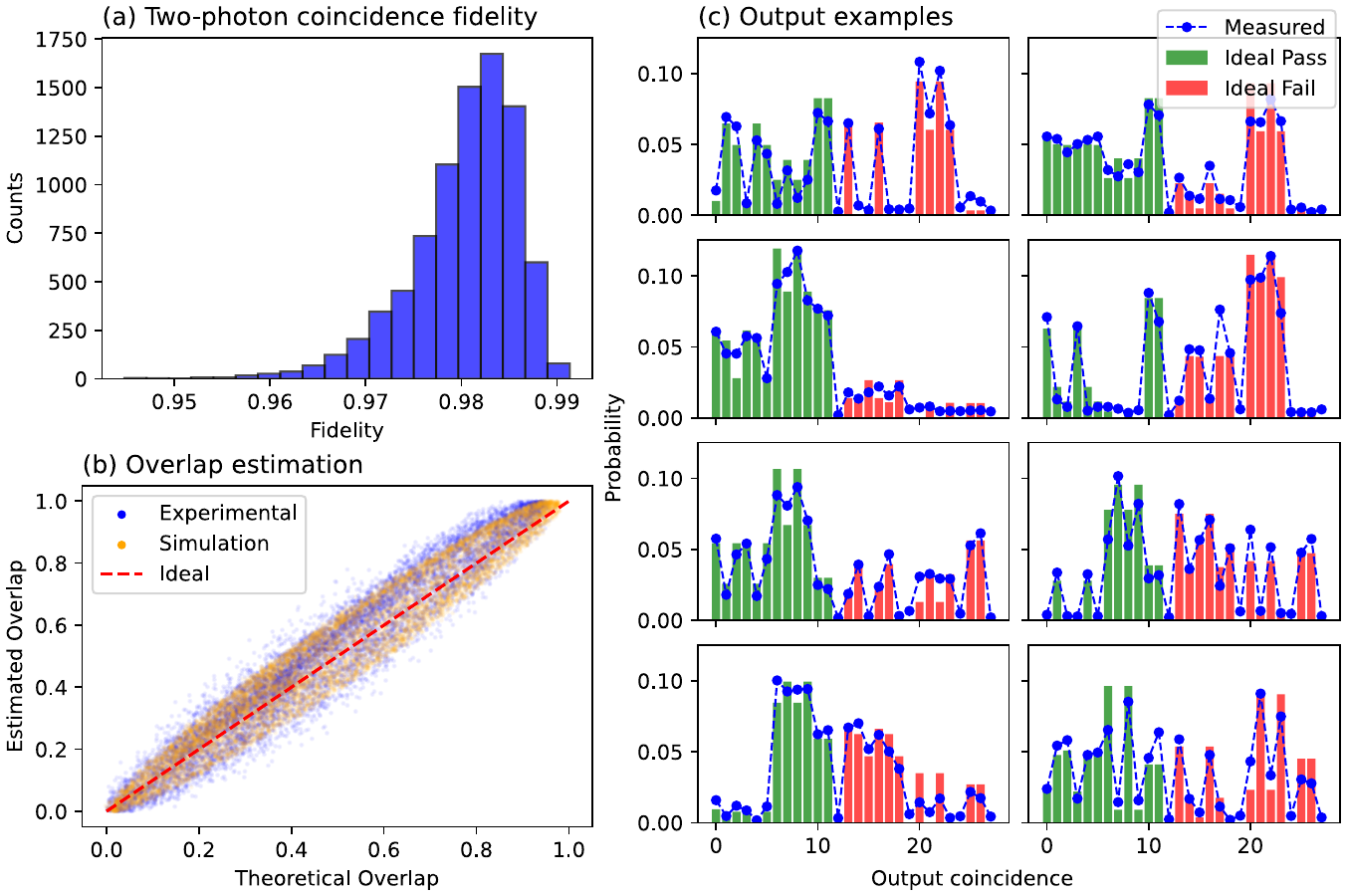}
    \caption{\label{supp fig: chip output} (a) Fidelity of the output coincidence distribution. Data in this figure are taken from the classification experiment for overlapping blobs dataset. (b) Scatter plot of experimentally estimated and simulated overlap values against theoretical overlap values. (c) Example output coincidence distributions for eight random circuit encodings. Green bars represent even-parity coincidence outcomes, denoted `Pass'. Red bars represent odd-parity coincidence outcomes, denoted `Fail'. Blue dashed line plots the experimentally measured output distribution. }
\end{figure*}

\subsection{Simulator}
We also designed a simple simulator that emulates the thermal crosstalk between phaseshifters. The single-phaseshifter sweeps in the previous section returns a phase-current relationship. If we set a phase according to this relationship, the phase value would be proportional to the local electric heat, but does not take into account any heat dissipation to neighbouring phaseshifters. If we denote the phase we set onto the phaseshifter $(i,j)$ according to single phaseshifter calibration result as $\tilde{\theta}_{i,j}$, and the actual phase change induced on the phaseshifter as $\theta_{i,j}$, then we use a simple mathematical relationship to emulate the crosstalk effect:
\begin{equation}\label{supp eqn: crosstalk model}
    \theta_{i,j}
    = 
    \tilde{\theta}_{i,j} + 
    \left(
    \sum_{l,m=0}^{9}
    K_{(i,j), (l.m)} 
    \tilde{\theta}_{l,m} \right)
    \left(
    1 + \eta_{i,j} \tilde{\theta}_{i,j}
    \right)
    + \epsilon_{i,j},
\end{equation}
where $K_{(i,j),(l,m)}$ is a coefficient that quantifies the proportion of heat leaked from phaseshifter $(l,m)$ to phaseshifter $(i,j)$, $\eta_{i,j}$ is a coefficient that we use to capture any non-linear crosstalk effects, and $\epsilon_{i,j}$ is a calibration error term. 

Theoretically, the summation term in Equation~\ref{supp eqn: crosstalk model} should be over all phaseshifters $(l,m)$. Practically, the $K$ coefficients should decay according to geometric distance. We model the coefficients as follows: 
\begin{equation}
K_{(i,j),(l,m)} 
= 
\begin{cases}
    & K^{(1)} \xi_{(i,j),(l,m)} 
    \quad 
    \text{if nearest neighbour}, \\
    & K^{(2)} \xi_{(i,j),(l,m)} 
    \quad 
    \text{if second-nearest neighbour}, \\
    & K^{(3)} \xi_{(i,j),(l,m)} 
    \quad 
    \text{if third-nearest neighbour}, \\
    & K^{(4)} \xi_{(i,j),(l,m)} 
    \quad 
    \text{otherwise}, \\
\end{cases}
\end{equation}
where $K^{(1,2,3,4)}$ are the baseline crosstalk coefficients for nearest, second-nearest, third-nearest neighbour (according to the phaseshifters' geometric layout on chip) and global crosstalk respectively, and $\xi$ is a normal distributed random value with mean $1$ and standard deviation $0.05$, i.e. $\xi\sim \mathcal{N}(1, 0.05)$. This simplifies the number of parameters in our simulator while also keeping variance between phaseshifters. 

Experimental tests revealed that the crosstalk behaviour on our chip between two phaseshifters, $(l,m)$ where the heat is initially applied to (or `aggressor phaseshifter'), and $(i,j)$ where the heat dissipates to (or `victim phaseshifter'), not only depends on $\tilde{\theta}_{l,m}$ of the aggressor phaseshifter, but also on the local $\tilde{\theta}_{i,j}$ of the victim phaseshifter. This is why an additional $\eta_{i,j}$ coefficient is included to capture this non-linear behaviour. Typically, this coefficient is very small but non-negligible. This adds difficulty to the crosstalk calibration step described in the previous section: we have to sweep neighbouring MZIs at sufficient points in their $2\pi$ range in order to construct a look-up table that fully captures this non-linear behaviour. In the digital simulator, we model the coefficients as a Gaussian random number with mean $\eta$ and standard deviation $0.1 \eta$, or $\eta_{i,j}\sim \mathcal{N}(\eta, 0.1\eta)$, where $\eta$ is a common non-linear coefficient for all phaseshifters. 

Finally, we model the calibration error as an additive error on the actual phase. We also set a global calibration error value, $\epsilon$, and for each phaseshifter $(i,j)$, the calibration error $\epsilon_{i,j}$ is a Gaussian random value with mean $\pm\epsilon$ (with equal probability) and standard deviation $0.01$ radians.  

Therefore, to summarise, this simple crosstalk simulator has the following hyperparameters we can tune: $K^{(1,2,3,4)}$ for baseline crosstalk coefficients, $\eta$ for baseline non-linear coefficient and $\epsilon$ for baseline calibration error. The simulator evolves the phase values we set without considering crosstalk, $\tilde{\theta}_{i,j}$, to phase values we would get if crosstalk behaves according to Equation~\ref{supp eqn: crosstalk model}. We can then simulate the output probabilities according to the input states. After rough tuning of the hyperparameters, we were able to simulate a spread of overlap errors that is similar to the experimental outcome, as shown by the orange scatter plots in Fig.~\ref{supp fig: chip output}(b). This involves setting $K^{(1)}=1.6\%, K^{(2)}=0.4\%, K^{(3)}=0.16\%, K^{(4)}=0.05\%,\eta=1\%$, and $\epsilon=0.02$~rad. This is the same simulator used to simulate the online learning protocol in the main text.

The crosstalk parameters are very low compared to the actual crosstalk behaviour we measured, which revealed typical values of $K^{(1)}=8\%, K^{(2)}=5\%$. This demonstrates the extent of success of our crosstalk mitigation method described in the previous section. However, the fact that we still have a non-negligible spread in the overlap values, as well as a limited fidelity in the output coincidence distribution, shows how significant the impact of thermal crosstalk is. 

Recently, machine-learning assisted routines have been proposed for high-fidelity photonic circuit calibration~\cite{fyrillasScalableMachineLearningassisted2024, youssryExperimentalGrayboxQuantum2024}. One direction for future improvement is adopting similar routines to directly learn the parameters in Equation~\ref{supp eqn: crosstalk model}. Then we can directly infer the relationship between $\tilde{\theta}$, the phase values we set, and $\theta$, the phase values we would actually get. Alternatively, black-box models such as neural networks or tree-based methods may also be adopted instead of the clear-box model we propose here. We leave this future work.

\section{Classification by support vector machines}\label{supp section: classification task}

Firstly, we introduce the classical support vector machine (SVM) algorithm.
Consider a set of data, $\{\bm{x}_i, y_i\}_{i=1}^m$, that consists of $m$ samples of $d$-dimensional data vectors, $\bm{x}_i\in \mathds{R}^{d}$, to each of which a binary label, $y_i=\pm1$, is assigned. In this supplementary section, we use subscript $i$ to label the different data instances, not to be confused with the mode label in earlier text. For now, we assume that the data are linearly separable in Euclidean space. In other words, we can find a hyperplane defined by some vector $\bm{w}\in \mathds{R}^d$ and some intercept $b\in \mathds{R}$, such that $y_i = \sign(\bm{w}^T\bm{x}_i + b)$, or that $y_i(\bm{w}^T\bm{x}_i + b) - 1 \geq 0$ for all $i\in[1, m]$. The goal of SVM is to find some $\bm{w}$ that maximises the separation between the two classes of data~\cite{cortesSupportvectorNetworks1995}. A textbook introduction of the algorithm can be found in e.g. Ref.~\cite{hastieElementsStatisticalLearning2009}.

To maximise the separation, the classification task is turned into a constraint optimisation problem: 
\begin{equation}\label{supp eqn: constraint optimisation for linear svm}
\begin{split}
    &\min_{\bm{w},b, \xi_i} \frac{1}{2} \|\bm{w}\|_2
    + C\sum_{i=1}^m \xi_i
    \\
    \text{subject to} \quad
    & \xi_i \geq 0
    \\
    &y_i(\bm{w}^T\bm{x}_i + b) \geq 1 - \xi_i 
    \; \forall \, i \in [1,m],
\end{split}
\end{equation}
where we introduced a soft-margin constraint with some slack constant $C$. A large $C$ corresponds to a lower tolerance for data points to fall on the wrong side of the decision boundary. A hard constraint, with all $\xi_i=0$, corresponds to $C=\infty$. The constrained optimisation problem can be solved by solving its Lagrangian dual program: 
\begin{equation}\label{supp eqn: dual program for linear svm}
\begin{split}
    &\max_{\beta_i}
    \sum_{i=1}^m \beta_i 
    -
    \frac{1}{2} \sum_{i=1}^m \sum_{j=1}^m 
    \beta_i \beta_j y_i y_j \bm{x}_i^T \bm{x}_j
    \\
    \text{subject to} \quad
    & 0 \leq \beta_i \leq C \\
    & \sum_{i=1}^m \beta_i y_i = 0. 
\end{split}
\end{equation}

To derive Equation~\ref{supp eqn: dual program for linear svm} from Equation~\ref{supp eqn: constraint optimisation for linear svm}, we first consider the hard-constraint problem with $\xi_i=0$. Taking Lagrange multipliers, $\beta_i>0$, constrained optimisation of Equation~\ref{supp eqn: constraint optimisation for linear svm} is equivalent to minimising the following Lagrangian without constraints: 
\begin{equation}
    \mathcal{L} = 
    \frac{1}{2} \bm{w}^T \bm{w}
    -
    \sum_{i} \beta_i 
    \left[y_i (\bm{w}^T \bm{x}_i + b) -1
    \right].
\end{equation}
Taking derivatives, we find, 
\begin{align}
    \frac{\partial \mathcal{L}}{\partial \bm{w}}
    &=
    \bm{w} - \sum_{i}\beta_i y_i \bm{x}_i = 0,
    \\
    \frac{\partial \mathcal{L}}{\partial b}
    &=
    - \sum_{i}\beta_i y_i = 0,
\end{align}
which we substitute back into the Lagrangian to recover the optimisation problem in Equation~\ref{supp eqn: dual program for linear svm}. Finally, if there exists some point
, $\bm{x}_i$, that cannot satisfy the hard constraint, they would theoretically have infinite support in $\beta_i = \infty$. In this case, we can clip their support at some slack constant, $C$, which recovers the soft margin of $0 \leq \beta_i \leq C$ in Equation~\ref{supp eqn: dual program for linear svm}.

A key observation is that the optimisation problem now \textit{only} depends on the inner product, $\bm{x}_i^T \bm{x}_j$, not on any individual data vectors. We can define this as a kernel function, $K(\bm{x}_i, \bm{x}_j)=\bm{x}_i^T \bm{x}_j$. If the data are no longer linearly separable in the original Euclidean space, instead of computing a non-linear classification problem, the `kernel trick' is to map the data vectors into a higher-dimensional space by defining a suitable alternative kernel function and then to solve the same linear problem in the new space. 

In our work, we adopt a quantum feature map, where three-dimensional unseen vectors specify phase-encoded qudits $\bm{x}\rightarrow |\psi(\bm{x})\rangle$, such that the kernel function becomes $K(\bm{x}_i, \bm{x}_j)=\left|\langle\psi(\bm{x}_i)|\psi(\bm{x}_j)\rangle\right|^2$. The data vectors $\bm{x}$ represent the three phase degrees of freedom of the qudit, $\bm{x}=\bm{\theta}$ (Equation 3 in main text). For notational consistency, we continue to use subscript-labelled $\bm{x}_i$ throughout this supplementary section, though in the main text we used superscripts instead to differentiate the data labels from the optical mode labels.

We denote the kernel function evaluations between the training data as an $m\times m$ overlap matrix $\bm{K}$, which has elements $K_{ij}=K(\bm{x}_i, \bm{x}_j)$, and the training labels and classifier weights as $m$-dimensional column vectors: $\bm{y}=(y_1, \hdots y_m)^T$, $\bm{\beta}=(\beta_1, \hdots, \beta_m)^T$. Then we can re-express the optimisation problem as the following quadratic program: 
\begin{equation}
\begin{split}
    &\min_{\bm{\beta}} 
    \frac{1}{2}\bm{\beta}^T 
    \left(\bm{y}^T \bm{K} \bm{y}\right)
    \bm{\beta} 
    + \bm{q}^T \bm{\beta}
    \\
    \text{subject to} \quad
    & \bm{G} (\bm{\beta}\oplus \bm{\beta}) \leq \bm{h} \\ 
    & \bm{y}^T \bm{\beta} = \bm{0},
\end{split}
\end{equation}
where vector $\bm{q}$ is a $m$-dimensional vector of $-1$s, i.e. $\bm{q}=(-1, \hdots, -1)^T$; matrix $\bm{G}$ is a direct sum of the $m\times m$ identity matrix and its additive inverse, i.e. $\bm{G}=\bm{I}_m\oplus (-\bm{I}_m)$; 
and vector $\bm{h}$ is a $2m$-dimensional vector where the first $m$ elements are $C$ and the last $m$ elements are $0$, i.e. $\bm{h}=(C, \hdots, C, 0, \hdots, 0)^T$. 
We then define a correction bias term $b$, given by 
\begin{equation}
    b = \frac{1}{m}\sum_{j=1}^m\left(y_j-\sum_{i=1}^m\beta_iy_iK(\bm{x}_j,\bm{x}_i) \right).
\end{equation}

In all three classification tasks shown in the main text, we train the classifier with $m=100$ training data vectors and we set the slack constant to $C=0.8$. After we measure the overlap matrix $\bm{K}$, we solve the quadratic program using the CVXOPT package in Python~\cite{cvxopt}. 

Finally, for unseen new data $\bm{x}$, we compute its kernel function with the training dataset and compute its label by
\begin{equation}
    \hat{y}(\bm{x}) = 
    \sign \left( 
    \sum_{i=1}^m \beta_i y_i K(\bm{x}, \bm{x}_i)+b
    \right). 
\end{equation}

Practically, the $\beta_i$ coefficients are non-zero only for a subset of the training data. The data vectors corresponding to non-zero $\beta_i$ are called \textit{support vectors} as they lie on the decision boundary of the hyperplane. The kernel function only needs to be evaluated against this subset of support vectors.

The quantum datasets used in our experiments are firstly generated in the parameter space of $\boldsymbol{\theta}$, i.e., the phases of the quantum states, while the amplitudes $A_i$ in Equation~\ref{supp eqn: qudit amplitude values} serve as hyperparameters. 
Through the scikit-learn package in python~\cite{JMLR:v12:pedregosa11a}, we generate three types of datasets: separate, spherical and overlapping.
Each dataset contains 100 training and 100 testing data vectors with proper parameter ranges.
We check these datasets by running classical SVM algorithm, in which the kernels are simulated using different choices of $\{A_i\}_{i=0}^{3}$.
Further, we optimize the hyperparameters to find the set $\{A_i\}_{i=0}^{3}$ in Equation~\ref{supp eqn: qudit amplitude values} that yields optimal simulated classification accuracies of $100\%$, $100\%$ and $97\%$ for separate, spherical and overlapping datasets, respectively.

\section{Online learning by simultaneous perturbation stochastic approximation}

\begin{figure*}[h]
    \includegraphics[width=0.9\textwidth]{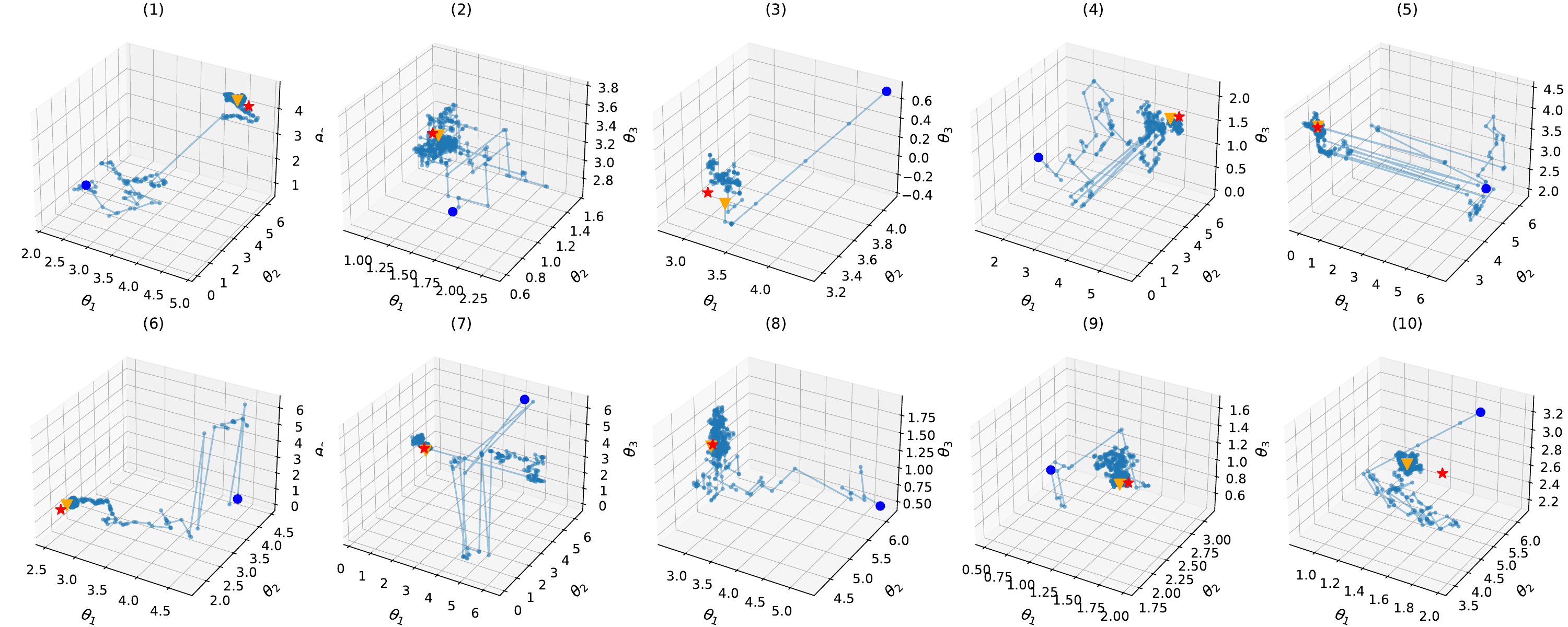}
    \caption{\label{supp fig: spsa all paths} The optimisation results for $10$ randomly selected target phases. Each panel represents one experiment. Data shown for $N=100$ samples per overlap measurement. The red star represents the phase $\bm{x}^{(t)}$ of the target state. The blue circle represents the initial guess, $\bm{\theta}^{(0)}$. The orange triangle represents the final guess after $500$ iterations of the SPSA algorithm. }
\end{figure*}

In our work we implemented an online learning protocol by simultaneous perturbation stochastic approximation (SPSA) algorithm~\cite{spall1998-overview-spsa, spall1998-spsa-implementations}. In this protocol, the unknown target state $|\psi(\bm{x}^{(t)})\rangle$, determined by encoded phases $\bm{x}^{(t)}=[\theta_1^{(t)}, \theta_2^{(t)}, \theta_3^{(t)}]$, is prepared on one of the two qudits.
Then, we iteratively update the phase of the other qudit, $|\psi(\bm{\theta})\rangle$, to maximise the overlap between the two qudits. When the overlap approaches unity, then the second qudit approximately recovers the target state $|\psi(\bm{x}^{(t)})\rangle$ that we wish to learn.
We define the cost function as 
\begin{equation}
    c(\bm{\theta}) = 1 - \left|\langle \psi(\bm{x}^{(t)})|\psi(\bm{\theta})\rangle\right|^2.  
\end{equation}

In the $0$-th iteration, we pick an initial guess $\bm{\theta}^{(0)}=[\theta_1^{(0)}, \theta_2^{(0)}, \theta_3^{(0)}]$, where each $\theta_i^{(0)}$ is randomly initialised between $0$ and $2\pi$. In the $k$-th iteration, we generate a random perturbation vector $\bm{\Delta}^{(k)} \in \{\pm1\}^{\otimes 3}$, where each component can be $\pm 1$ with equal outcome (i.e. follows the Bernoulli distribution). We evaluate the cost function at the simultaneous perturbation around the current phase values, $\bm{\theta}^{(k)}$: $c(\bm{\theta}^{(k)}\pm \bm{\Delta}^{(k)} t^{(k)})$, where $t^{(k)}$ is a perturbation coefficient we will define later. Operationally, this means encoding the phases $\bm{\theta}^{(k)}\pm \bm{\Delta}^{(k)} t^{(k)}$ and measuring the overlaps respectively. We repeat the experiments for $N=100, 1000$ and $10000$ samples per overlap measurement, as shown in the main text. 

The simultaneous perturbation allows us to approximate the local gradient at $\bm{\theta}^{(k)}$ to be:
\begin{equation}
    \bm{g}(\bm{\theta}^{(k)}) = 
    \frac{
    c\left(\bm{\theta}^{(k)} + \bm{\Delta}^{(k)} t^{(k)}\right) 
    - c\left(\bm{\theta}^{(k)} - \bm{\Delta}^{(k)} t^{(k)}\right)
    }{
    2 t^{(k)} \bm{\Delta}^{(k)}
    }.
\end{equation}
We then update the phase parameters for the $k+1$-th iteration to be 
\begin{equation}
    \bm{\theta}^{(k+1)} = \bm{\theta}^{(k)} 
    - a^{(k)} \bm{g}(\bm{\theta}^{(k)}),
\end{equation}
where $a^{(k)}$ is a gain coefficient that determines how much we update the phase parameters at each iteration. 

The perturbation and gain coefficients decay with iterations, and are defined as 
\begin{equation}
    a^{(k)} = \frac{a}{(A+k+1)^\alpha}, 
    \quad 
    t^{(k)} = \frac{t}{(k+1)^\gamma}. 
\end{equation}
The coefficients $\alpha, \gamma$ govern the rate of decay for both $a^{(k)}$ and $t^{(k)}$. Following the guidelines of Ref.~\cite{spall1998-spsa-implementations}, we choose $\alpha=0.602, \gamma=0.101$ respectively. The coefficient $A$ is a `stability constant' that allows for larger $a$ coefficient in the numerator of $a^{(k)}$ while ensuring stability in early iterations. Ref.~\cite{spall1998-spsa-implementations} recommends to choose a value of $A$ that is much less than the maximum number of iterations allowed. In our experiments, we run the protocol for $500$ iterations, and choose $A=10$. The $a$ coefficient determines the step size of each iteration. Heuristically, we found that setting $a=1.6$ gives a good convergence rate for our experiments. 

Finally, the $t$ coefficient determines the magnitude of the perturbation at each step. Ref.~\cite{spall1998-spsa-implementations} recommends a selection of $t$ that is approximately the standard deviation of the measurement noise, so that the perturbed cost function values are noise-robust. We follow this guideline. At the start of the $0$-th iteration, we evaluate the initial cost function, $c(\bm{\theta}^{(0)})$, five times and set $t$ to be twice the standard deviation of the measurement results. 

We repeat the optimisation routine for ten randomly selected target phases, corresponding to ten different target states. The results are summarised in the main text as well as Fig.~\ref{supp fig: spsa all paths}. 

\end{document}